%% file: main.tex
\documentclass[sigconf,nonacm]{aamas}
\usepackage{balance}

\makeatletter

\if@ACM@newfonts\else
  \RequirePackage{lmodern}
  \RequirePackage{amssymb}
\fi
\gdef\@copyrightpermission{
  \begin{minipage}{0.2\columnwidth}
   \href{https://creativecommons.org/licenses/by/4.0/}{\includegraphics[width=0.90\textwidth]{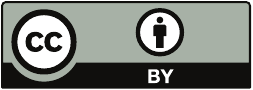}}
  \end{minipage}\hfill
  \begin{minipage}{0.8\columnwidth}
   \href{https://creativecommons.org/licenses/by/4.0/}{This work is licensed under a Creative Commons Attribution International 4.0 License.}
  \end{minipage}
  \vspace{5pt}
}
\makeatother

\usepackage{mathtools}
\usepackage{booktabs,tabularx,multirow}
\usepackage{tikz}
\usetikzlibrary{arrows.meta,positioning,fit,calc}
\usepackage{enumitem}
\usepackage{xspace}
\usepackage{microtype}

\newtheorem{definition}{Definition}
\newtheorem{theorem}{Theorem}
\newtheorem{proposition}{Proposition}
\newtheorem{corollary}{Corollary}
\newtheorem{lemma}{Lemma}
\newtheorem{assumption}{Assumption}
\newtheorem{example}{Example}

\newcommand{\RRI}{\textsc{RRI}\xspace}
\newcommand{\CRI}{\textsc{CRI}\xspace}
\newcommand{\LNI}{\textsc{LNI}\xspace}
\newcommand{\FLR}{\textsc{FLR}\xspace}
\newcommand{\QFLR}{\textsc{qFLR}\xspace}
\newcommand{\dfquad}{DF-QuAD\xspace}
\newcommand{\lqdf}{LQ-DF-QuAD\xspace}
\newcommand{\profiles}{\mathcal{P}}

\newcommand{\uR}{u_{R}}
\newcommand{\repr}{\varphi}
\newcommand{\ql}{q_{\!L}}
\newcommand{\qc}{q_{\!C}}
\newcommand{\qk}{q_{\!K}}
\newcommand{\rootagg}{\oplus}
\newcommand{\multisetplus}{\mathbin{\uplus}}
\newcommand{\score}{S}
\newcommand{\decision}{D}
\newcommand{\equivL}{\equiv_{R}}
\newcommand{\approxL}{\simeq_{R}}

\title{Replication or Corroboration? An Axiomatic Framework for Evidential Dependence in Multi-Agent Debate}

\author{Hoang-Loc Cao}
\authornote{These authors contributed equally to this work.}
\affiliation{%
  \institution{University of New Brunswick}
  \country{Canada}
}
\email{chloc22@clc.fitus.edu.vn}

\author{Truong Thanh Hung Nguyen}
\authornotemark[1]
\affiliation{%
  \institution{University of New Brunswick}
  \country{Canada}
}
\email{hung.ntt@unb.ca}

\author{Phuc Truong Loc Nguyen}
\affiliation{%
  \institution{FAU Erlangen-N\"urnberg}
  \country{Germany}
}
\email{loc.pt.nguyen@fau.de}

\author{Vo Thanh Khang Nguyen}
\affiliation{%
  \institution{FPT Software}
  \country{Vietnam}
}
\email{khangnvt1@fpt.com}

\author{Phuc Ho}
\affiliation{%
  \institution{University of New Brunswick}
  \country{Canada}
}
\email{hdphuc22@clc.fitus.edu.vn}
\author{Hung Nguyen}
\affiliation{%
  \institution{Vietnamese-German University}
  \country{Vietnam}
}
\email{102240004@student.vgu.edu.vn}
\author{Duc Pham}
\affiliation{%
  \institution{HCM University of Technology}
  \country{Vietnam}
}
\email{duc.pham0609@hcmut.edu.vn}

\author{Long Tran-Thanh}
\affiliation{%
  \institution{University of Warwick}
  \country{United Kingdom}
}
\email{long.tran-thanh@warwick.ac.uk}

\author{Hung Cao}
\affiliation{%
  \institution{University of New Brunswick}
  \country{Canada}
}
\email{hung.cao@unb.ca}

\renewcommand{\shortauthors}{Cao et al.}

\begin{document}

\begin{abstract}
In multi-agent debates, matching reports can arise for two different reasons: agents may repeat the same evidence (replication), or they may provide separate evidential roots that support the same claim (corroboration). A judge must distinguish these two cases. Counting report occurrences gives extra weight to relays, paraphrases, alias resubmissions, and excerpt splits, while merging all matching content avoids this only by also removing genuine corroboration. We formalize this tension as a minimal judgment game in which acquiring new evidence is costly, but replication is cheap. We show that \emph{root replication invariance} (\RRI) is exactly threshold-uniform replication immunity. Under \RRI, replication-only deviations cannot change the verdict, and positive-cost replication is strictly dominated, whereas occurrence-sensitive and content-keyed judges admit profitable deviations. We also establish an information boundary: if a judge's observation interface cannot distinguish a replica from a matched fresh root, no judge can both ignore the replica and respond strictly to the fresh root. To preserve this distinction, we introduce the \emph{lineage quotient}, which removes wording, identity, occurrence multiplicity, and arbitrary root names while keeping distinct roots separate. On orbit-connected domains, it is the coarsest representation sufficient for every deterministic judge invariant to replication and root-label renaming. We further lift this construction to quantitative bipolar argumentation frameworks (QBAFs), characterize strict fresh-root response, prove its propagation along sign-consistent paths under DF-QuAD, and derive a tight closed-form response modulus. Exhaustive audits over 28,224 direct-star profiles, 2,052 multi-layer DAGs, with matched interventions on 277 AVeriTeC-derived profiles, show the predicted separation between replication and corroboration under a supplied partition.

\end{abstract}


\maketitle
\pagestyle{plain}
\thispagestyle{plain}

\section{Introduction}
A multi-agent debate protocol, whether through an external judge or the agents' own aggregation step, must decide how much weight to give the reports it receives \cite{irving2018debate,du2024debate,chan2024chateval}. The judge sees who submitted each report, what it says, and which claim it supports or attacks, but these observations do not reveal whether reports that agree add distinct evidence. Source-dependence and truth-discovery work has long observed that apparent corroboration can arise from copying \cite{dong2009dependence,li2016survey}, and jury-theorem and opinion-pooling analyses make the probabilistic version of the same point \cite{ladha1992condorcet,dietrich2004jury,dietrich2013independent,genest1986combining}. Thus, report counts and textual agreement are an incomplete basis for judgment, and in domains that rely on source integrity, such as automated fact-checking and claim verification, confusing replication with corroboration can inflate confidence in unverified or false claims.
Consider three agents who submit the same supporting claim. They may have found three independent records, or they may be passing along a single record that has been relayed, paraphrased, excerpted from another passage of the same document, or resubmitted under a second displayed identity. We call the underlying piece of evidence, such as one record or witness statement, an \textit{evidential root} (\textit{root} for short), and each copy, paraphrase, excerpt, or resubmission of it a \textit{rendering} of that root. A root is \textit{fresh} if no report submitted so far renders it: each independent record is a fresh root when it first arrives, whereas relaying a record that has already been submitted adds a rendering but no fresh root (formalized in Sec.~\ref{sec:units}).

\begin{figure*}[ht]
    \centering
    \input{fig/fig_overview} 
    \caption{\textbf{Method overview.} (a)~Replication and corroboration become indistinguishable once root labels are dropped (Thm.~\ref{thm:information}). (b)~Only the lineage quotient $q_L$ ignores replicas and excerpt splits yet counts fresh roots, and it is the coarsest such key (Thms.~\ref{thm:factor}--\ref{thm:minimal}). (c)~Lifting $q_L$ to QBAFs preserves these guarantees for any semantics (Thm.~\ref{thm:lifting}), making replication strictly dominated while fresh roots still propagate. Color denotes evidential root.}
    \label{fig:overview}
\end{figure*}

\paragraph{Why Replication is a Multi-Agent Problem.}
In a debate, duplication is a strategic choice available to agents. An agent can relay a peer's evidence, paraphrase it, quote a different passage of the same document, or resubmit it under a second identity without acquiring anything new. The judge's representation therefore determines whether agents can gain influence simply by repackaging existing evidence. Sec.~\ref{sec:model} defines this dilemma precisely through a minimal judgment game in which acquiring a fresh evidential root costs strictly more than any replication action. In this game, neutrality to further renderings of a represented root is \emph{equivalent} to immunity of the threshold verdict to the entire replication class at every threshold (Prop.~\ref{prop:repproof}). Consequently, no coalition profits from a replication-only deviation, and positive-cost replication is strictly dominated (Cor.~\ref{cor:repproof}). We call a judge with this property \emph{coalition replication-proof}, since repeating or re-rendering represented evidence cannot change the outcome, regardless of which coalition performs the replication or where the additional reports are inserted. Yet collapsing matching content can discard genuine corroboration, undermining a central benefit of debate. The representation must therefore remove the incentive to replicate while preserving the value of acquiring distinct evidence.

Now, a judge that aggregates every occurrence gives relays and copies the same weight as independent evidence. On the other hand, a judge that merges every matching proposition avoids that error, but then prevents reports from distinct roots from reinforcing one another whenever their content happens to agree. Now, consider a third response option, which also fails, but less visibly, and motivates the key idea used in our paper. A judge that merges reports agreeing in both content and root ignores verbatim copies, yet still rewards an agent who submits two excerpts of one document, because the excerpts differ in content (Ex.~\ref{ex:deviation}). The distinction must therefore be drawn at the level of \emph{roots}. Further renderings of a represented root should leave the score unchanged, while an eligible contribution from a distinct root should be able to strengthen support or attack. Fig.~\ref{fig:overview}a) illustrates this distinction. It leads to the central question of the paper: \emph{What must a judge retain to separate replication from genuine corroboration?}

\paragraph{Trust Boundary}
To answer this question, we take lineage as provided, with the judge receiving a partition of report occurrences into represented evidential roots. The judge does not infer this partition, and distinct roots do not certify causal or statistical independence, since dependence inference and identity authentication remain upstream tasks \cite{dong2009dependence,li2016survey}. All renderings of a single root receive a single label, so an agent cannot increase a root's aggregation weight by fragmenting its content. If one document carries several separable contributions, the upstream model may record it as several roots, and that decision is explicit and open to audit. An agent cannot force the same split by choosing what to excerpt.
Our guarantees therefore concern aggregation given supplied dependence information, not provenance discovery or fact-checking accuracy.

Within this trust boundary, we introduce three principles that can lead to the required behavior: 

(1) \emph{Root Replication Invariance} (\RRI): This requires that a further rendering of a represented root, no stronger than what that root has already contributed, leaves the score unchanged. Restricting \RRI to repetitions with identical content yields the weaker property, which we call \textit{Content-keyed Replication Invariance} (\CRI). In our work, two reports count as such a repetition when they agree in content, root, argumentative role, and strength. 

(2) \emph{Fresh-Lineage Responsiveness} (\FLR): This requires that an eligible new supporter or attacker moves the score in the correct direction, and its quantitative form \QFLR sets a lower bound on how far the score moves. 

(3) \emph{Lineage-Name Invariance} (\LNI): This requires that the spelling of root label carries no information, which still allows provenance-aware weighting by source, provided that source attributes are represented as explicit occurrence fields.

These principles raise a natural chain of questions: what must a judge observe, which representation retains exactly that information, and how far does a fresh root move the score?. Our contributions follow that chain and are summarized as follows:

\textbf{1. Replication Immunity and Information Requirements} (Sec.~\ref{sec:model}--\ref{sec:principles}): We introduce a minimal judgment game and show that \RRI is equivalent to threshold-uniform replication immunity (Prop.~\ref{prop:repproof}). Under \RRI, replication-only deviations cannot change the verdict, and positive-cost replication is strictly dominated (Cor.~\ref{cor:repproof}), while occurrence aggregation, plurality voting, and a content-keyed quotient each admit a strictly profitable deviation (Ex.~\ref{ex:deviation}). We also establish an information-necessity result showing that no observation interface can ignore a replica while responding strictly to a fresh root if it represents the two identically (Thm.~\ref{thm:information}). This limitation applies to any observation interface, including a natural-language prompt given to a large language model (LLM).

\textbf{2. A Dependence-Aware Representation} (Sec.~\ref{sec:repr}--\ref{sec:qbaf}): We introduce the lineage quotient, a representation that merges renderings of the same root while keeping distinct roots separate, and that removes wording, identity, and occurrence multiplicity. On a fixed orbit-connected domain, a deterministic judge satisfies \RRI and \LNI exactly when its score depends on the profile only through this quotient (Thm.~\ref{thm:factor}). Any representation supporting the full class of such judges must therefore retain at least the information encoded by the lineage quotient (Thm.~\ref{thm:minimal}), making it the coarsest representation compatible with these requirements. We also show that content-based alternatives either remain vulnerable to replication or suppress genuine corroboration (Prop.~\ref{prop:contentcollapse}).

% (3) \textbf{Quantitative Response Guarantees}  (Sec.~\ref{sec:qbaf}): We characterize fresh-root response for signed-aggregative semantics, where response in the direct-leaf setting holds exactly when the relevant aggregate strictly increases (Thm.~\ref{thm:characterization}). Under \dfquad, this response propagates strictly to the topic whenever all directed paths from the leaf to the topic have the same attack-edge parity (Thm.~\ref{thm:propagate}), and a diamond with opposite parities demonstrates why this condition cannot be dropped. We also derive a tight closed-form modulus for \dfquad (Thm.~\ref{thm:modulus}), measuring how much corroboration can still move the score as the aggregate approaches saturation. Exhaustive structural audits (Sec.~\ref{sec:audits}) over 28,224 direct-star profiles and 2,052 multi-layer DAGs, together with a controlled audit on 277 AVeriTeC-derived profiles \cite{schlichtkrull2023averitec}, confirm the predicted separation between replication and fresh-root corroboration under a supplied partition.

\textbf{3. Quantitative Response Guarantees} (Sec.~\ref{sec:qbaf}): We characterize fresh-root response for signed-aggregative semantics, showing that in a direct-leaf setting, the response holds exactly when the relevant aggregate strictly increases (Thm.~\ref{thm:characterization}). Under \dfquad, we prove that this strict response propagates to the topic whenever all directed paths from the leaf to the topic share the same attack-edge parity (Thm.~\ref{thm:propagate}). Finally, we derive a tight closed-form modulus for \dfquad (Thm.~\ref{thm:modulus}), quantifying how much corroboration can still move the score as the aggregate approaches saturation.

\textbf{4. A Replication--Corroboration Stress Test} (Sec.~\ref{sec:audits}): We operationalize our theory through a matched-intervention audit that tests whether a judge ignores additional renderings of represented evidence while remaining responsive to genuinely fresh evidence. Exhaustive structural audits over 28,224 direct-star profiles and 2,052 multi-layer DAGs, alongside controlled interventions on 277 AVeriTeC-derived profiles \cite{schlichtkrull2023averitec}, expose the complementary failures of occurrence-sensitive and content-collapsing baselines. Crucially, these audits confirm the predicted empirical separation achieved by the lineage quotient.

\section{Related Work}
\label{sec:related}

\paragraph{Multi-Agent Debate and Post-hoc Judgment}
Multi-agent debate has been studied for scalable oversight and collaborative reasoning \cite{irving2018debate,khan2024persuasive,chan2024chateval,du2024debate,chen2024reconcile,liang2024divergent}, and its effectiveness depends on the protocol, agent diversity, and confidence handling \cite{smit2024mad,zhu2026demystifying}. LLM judges show systematic biases \cite{wang2024fair} and do not always capture non-linear argumentation structures \cite{sanayei2025judge}. The closest work is Yin \textit{et al.}~\cite{yin2026posthoc}, whose judge receives an \emph{agent-indexed opinion profile} that links each opinion to its agent but represents no relations between reports. By Thm.~\ref{thm:information}, such a judge can discount a repeated report only by speaker, not by evidential root, so an agent can inflate a report's weight at no cost by adding identities (action O4 in Sec.~\ref{sec:game}). Their setup models neither report replication nor false identities. Claim-verification systems also build claim--evidence graphs \cite{freedman2025argllm}. We ask how such graphs should represent reports that may share an evidential root.

\paragraph{QBAF Semantics and Duplicate Handling}
Quantitative argumentation extends graph-based reasoning \cite{dung1995acceptability,atkinson2017towards,cyras2021contestable} with graded strengths and signed relations \cite{baroni2015automatic,amgoud2018weighted,potyka2018continuous}, and has been applied to truth discovery \cite{potyka2024truth}. Gradual semantics such as \dfquad~\cite{rago2016dfquad} have been studied through their principles \cite{baroni2018properties,baroni2019principles} and through strength changes that explain outcome changes \cite{kampik2026strength}. Related notions include duplicate invariance \cite{gaggl2015reasonable}, aggregative semantics in which repeated information adds nothing while additional information reinforces \cite{munro2026aggregative}, and independence of clones in social choice \cite{tideman1987clones,elkind2012clones}. These semantics assume the graph is given, and duplicate invariance and clone independence suppress every near-copy. We instead build the graph from a supplied partition keyed by evidential root rather than report text, so renderings of one root add nothing while matching reports from distinct roots may still reinforce one another.

\paragraph{Provenance and Dependent Evidence}
Truth discovery estimates source trust, copying, and dependence \cite{yin2008truthfinder,dong2009dependence,li2016survey}, and trust- and provenance-aware argumentation uses source or derivation information to evaluate arguments \cite{koster2014trust,toniolo2014provenance}. Under dependent testimony, Condorcet jury results can weaken or reverse \cite{ladha1992condorcet,berend2007monotonicity}, and judgment aggregation \cite{dietrich2004jury,dietrich2013independent} and opinion pooling \cite{degroot1974consensus,genest1986combining} also treat shared evidence or dependent sources. These works identify or use dependence, or model how it changes the benefit of aggregation under a given probabilistic model. We instead ask which distinctions a representation must retain when dependence is given as a partition, and leave partial or graded dependence to Sec.~\ref{sec:discussion}.

\paragraph{Strategic Multiplicity}
Mechanism design for argumentation and group evaluation studies strategic reporting \cite{rahwan2008mechanism,awad2017group}, and false-name models consider agents with multiple identities \cite{todo2011falsename,brill2016falsename}. Our guarantee is narrower than false-name-proofness, covering only replication of already represented evidence under authenticated root labels, not forged labels or modified content.

\section{Model and the Judgment Game}
\label{sec:model}

\subsection{Occurrences, Roots, and Units}
\label{sec:units}
Let $c$ be the topic claim, $\mathcal C$ a finite set of structural claim nodes with $c\in\mathcal C$, $N$ a finite set of agents, and $H_i$ the displayed identities controlled by agent $i$. Let $\mathcal K$ be normalized propositions and $\mathcal L$ lineage labels, each naming one represented evidential root.

\begin{definition}[Report occurrence]
\label{def:occurrence}
An occurrence is $a=(h,x,\allowbreak\kappa,\lambda,p,s,\tau)$, where $h$ is a displayed identity, $x$ is surface text, $\kappa\in\mathcal K$ is normalized propositional content, $\lambda\in\mathcal L$ is a root label, $p\in\mathcal C$ is a targeted claim, $s\in\{+,-\}$ is the relation to $p$, and $\tau\in[0,1]$ is a strength. The controller of $h$ is not observed by the judge. A profile $P\in\profiles$ is a finite multiset of occurrences.
\end{definition}

Thus, paraphrases may have different surface texts \(x\) but the same \(\kappa\), whereas reports with the same \(\kappa\) may have either the same or distinct evidential roots \(\lambda\). Accordingly, the aggregation unit is the root together with the role it is used in, not the occurrence or its content.

\begin{definition}[Root--role unit, root strength, and fresh root]
\label{def:unit}
$\uR(a)=(\lambda(a),p(a),s(a))$. Write $\uR(P)$ for the units represented in $P$, and $a\equivL b$ iff $\uR(a)=\uR(b)$. The root strength of $\upsilon\in\uR(P)$ is $\rho_P(\upsilon)=\rootagg\{\tau(a)\mid a\in P,\uR(a)=\upsilon\}$, and $\Lambda(P)=\{\lambda(a)\mid a\in P\}$ is the set of roots \emph{represented} in $P$. A root $\ell\in\mathcal L$ is \emph{fresh for} $P$ if $\ell\notin\Lambda(P)$, and an occurrence $b$ is fresh for $P$ if $\lambda(b)$ is.
\end{definition}

Neutrality to additional renderings requires absorption,
$t\le\rootagg V\Rightarrow\rootagg(V\multisetplus\{t\})=\rootagg V$.
Together with $\rootagg\{t\}=t$, this uniquely yields
$\rootagg=\max$, which we use throughout.

\begin{assumption}[Root integrity]
\label{as:root}
One root label is assigned to every occurrence derived from one represented root, covering copies, paraphrases, alias resubmissions, and distinct excerpts of that root. Distinct labels denote roots that the upstream model represents as distinct. No agent action changes an occurrence's label.
\end{assumption}
\begin{assumption}[Occurrence strength]
\label{as:strength}
$\tau(a)\in[0,1]$ is a function of normalized propositional content, role, and explicitly represented source attributes. It does not depend on surface text, identity, root label, or occurrence count. The root strength $\rho_P=\max$ is bounded by the root ceiling $\rho^\ast(\upsilon)=\sup\{\tau(a):\uR(a)=\upsilon\}$.
\end{assumption}

\begin{assumption}[Role congruence]
\label{as:role}
Role is part of the unit key, so profile-level congruence is automatic; Sec.~\ref{sec:qbaf} states the graph-level form for general QBAFs.
\end{assumption}

\subsection{Agent Actions and Minimal Judgment Game}
\label{sec:game}
At each step an agent $i\in N$ either abstains, \textbf{(O0)}, or appends one occurrence to $P$: \textbf{(O1)} \emph{fresh contribution}, a root outside $\Lambda(P)$; \textbf{(O2)} \emph{relay}, another occurrence of a unit already in $\uR(P)$ with the same content, strength, and identity; \textbf{(O3)} \emph{paraphrase}, as (O2) with different surface text; \textbf{(O4)} \emph{alias resubmission}, as (O2) under a different identity $h'\in H_i$; \textbf{(O5)} \emph{excerpt split}, a different canonical proposition for a unit already in $\uR(P)$, at strength $\tau\le\rho_P(\upsilon)$; or \textbf{(O6)} \emph{strength revelation}, as (O5) but with $\tau>\rho_P(\upsilon)$, revealing better evidence from a known root. Let $\mathcal R=\{\text{O2,\ldots,O5}\}$ be the \emph{replication class}: none of its actions acquires new evidence, changes $\uR(P)$, or raises $\rho_P$, so none disturbs the admissibility of any other action. Action (O6) is outside $\mathcal R$, as intended behavior, bounded by the root ceiling $\rho^\ast$, and is not covered by the neutrality results below.

Agents know which root they render; the judge observes only the occurrence fields and never the controller of a displayed identity. Actions in $\mathcal R$ require no new evidence and cost accordingly, whereas (O1) requires a root the supplied dependence model recognizes as new. We make this a game.

\begin{definition}[Judgment game]
\label{def:game}
A judgment game is a tuple
\begin{equation}
    \mathcal G=\bigl(N,\{H_i\},\mathcal L,\{E_i\},\{c_i\},\{w_i,d_i\},J,\theta\bigr),
\end{equation}
where agent $i$ privately holds a root endowment $E_i\subseteq\mathcal L$ (endowments may overlap) and may submit label $\lambda\in E_i\cup\Lambda(P)$ at the current profile. An \emph{action list} $\alpha_i$ is a finite action sequence; a \emph{trace} interleaves $(\alpha_i)_{i\in N}$ admissibly and induces profile $P(\alpha)$. Costs are $c_i(\text{O0})=0$, $c_i(\text{O1})=c^f_i>0$, $c_i(\text{O2}\text{--}\text{O5})=c^r_i\in[0,c^f_i)$, $c_i(\text{O6})=c^v_i\ge0$, and utility is $U_i(\alpha)=w_i\mathbf 1[\decision_{J,\theta}(P(\alpha))=d_i]-c_i(\alpha_i)$ for target verdict $d_i\in\{0,1\}$ and stake $w_i>0$.
\end{definition}

Only the cost ordering $c^r_i<c^f_i$ is used below, and the neutrality results below hold for arbitrary nonnegative costs.

\begin{definition}[Replication deviation and replication-proofness]
\label{def:repdev}
$\alpha'_i$ is a \emph{replication-only deviation} from $\alpha_i$ if deleting $\mathcal R$-actions from both yields the same sequence; a \emph{coalition replication deviation} applies this to each $i\in M\subseteq N$. $(J,\theta)$ is \emph{coalition replication-proof} on domain $\mathcal D$ if $\decision_{J,\theta}$ is unchanged by every such deviation on every admissible trace in $\mathcal D$; a deviation is \emph{profitable} if it strictly increases $\sum_{i\in M}U_i$.
\end{definition}

\subsection{Matched Interventions}
\label{sec:matched}
\begin{definition}[Root replica and matched fresh extension]
\label{def:matched}
Fix $P$ and $a\in P$. A \emph{root replica} $b_R$ has $\uR(b_R)\in\uR(P)$ and $\tau(b_R)\le\rho_P(\uR(b_R))$, so $P_R=P\multisetplus\{b_R\}$ adds neither a unit nor strength. A \emph{matched pair at $a$} $(b_R,b_F)$ agrees in identity, text, content, target, sign, and strength, with $\lambda(b_R)=\lambda(a)$ and $\lambda(b_F)\notin\Lambda(P)$; the \emph{matched fresh extension} is $P_F=P\multisetplus\{b_F\}$.
\end{definition}
$P_R$ and $P_F$ differ in exactly one field: the supplied dependence annotation. An \emph{eligibility specification} is a pair $E^+,E^-\subseteq\profiles\times\profiles$, fixed before any score is compared, that may use structural well-formedness, freshness, role, strength, and a non-saturation test, but never the score inequality being tested; $E$ is \emph{non-vacuous} on $\mathcal D$ if some positive and some negative matched pair in $\mathcal D$ is eligible. A domain is \emph{$E$-lineage-rich at $P$} if it contains such a matched pair with $(P,P_F)\in E^+\cup E^-$.

\section{Principles and the Information Boundary}
\label{sec:principles}
A score judge is $J:\mathcal D\to[0,1]$ with score $\score_J(P)$; threshold $\theta$ induces $\decision_{J,\theta}(P)=\mathbf 1[\score_J(P)\ge\theta]$.

\begin{definition}[\RRI and \CRI]
\label{def:rri}
$J$ satisfies \RRI iff $\score_J(P\multisetplus\{b_R\})=\score_J(P)$ for every root-replica extension with both endpoints in $\mathcal D$. \CRI is the restriction of \RRI to replicas $b$ that agree with some $a\in P$ in content, root, target, sign, and strength.
\end{definition}
\RRI implies \CRI, and the converse fails on exactly the excerpt-split action, as shown in Ex.~\ref{ex:deviation}.

\begin{definition}[\LNI]
\label{def:lni}
$J$ satisfies \LNI iff every root-label bijection applied consistently leaves the score unchanged whenever both profiles lie in $\mathcal D$.
\end{definition}
\LNI requires invariance only to bijective renaming of root labels. It does not preclude provenance-aware weighting when the relevant source properties are represented explicitly as occurrence attributes used to determine $\tau$ (Asm.~\ref{as:strength}). Thus, \LNI removes dependence on arbitrary root names without restricting such weighting.

\begin{definition}[\FLR, \QFLR]
\label{def:flr}
Relative to fixed $E^+,E^-$, $J$ satisfies positive \FLR iff $\score_J(P')>\score_J(P)$ for every $(P,P')\in E^+$, and negative \FLR symmetrically on $E^-$. $J$ satisfies \QFLR with modulus $\delta$, a function of the strength of the fresh input, the relevant pre-intervention aggregate, and the topic base, iff the two score gaps are each at least $\delta>0$.
\end{definition}

In short, \RRI asks that repackaged evidence not matter, \LNI that root names not matter, \FLR that eligible fresh evidence move the score in its direction, and \QFLR that it do so by at least~$\delta$.

\subsection{Replication-Proofness Characterizes \RRI}
\label{sec:whyrri}
We call $(P^0,\ldots,P^k)$ a \emph{replication path} if each step is a root-replica extension. By Sec.~\ref{sec:game}, replication paths are exactly the profiles reachable by any interleaving of $\mathcal R$, performed by any agents through any number of identities.

\begin{proposition}[Replication-proofness characterises \RRI]
\label{prop:repproof}
For $J:\mathcal D\to[0,1]$: (i) $J$ satisfies \RRI; (ii) for every $\theta\in(0,1]$ and root-replica extension in $\mathcal D$, $\decision_{J,\theta}$ is unchanged; (iii) for every $\theta\in(0,1]$ and replication path in $\mathcal D$, $\decision_{J,\theta}(P^k)=\decision_{J,\theta}(P^0)$. These are equivalent.
\end{proposition}
\begin{proof}[Proof]
(i)$\Rightarrow$(iii) the score is constant along the path; (iii)$\Rightarrow$(ii) take $k=1$; (ii)$\Rightarrow$(i) if $\score_J(P)<\score_J(P_R)$, set $\theta=\score_J(P_R)\in(0,1]$ to get $\decision_{J,\theta}(P_R)=1\ne0=\decision_{J,\theta}(P)$, contradicting (ii); symmetric otherwise.
\end{proof}
The quantifier over $\theta$ is load-bearing: invariance at one \emph{fixed} deployed threshold is strictly weaker. At $v_0=\tfrac15$ with one supporter of strength $\tfrac15$ under occurrence \dfquad, one relay moves the score from $0.36$ to $0.488$, both below $\theta=\tfrac12$, so the verdict is stable at that threshold, although \RRI fails. We therefore describe \RRI as \emph{threshold-uniform replication immunity}.

\subsection{Information Necessity}
\label{sec:interfaces}
Let $\repr:\mathcal D\to\mathcal Z$ be any observation interface, $J$ is $\repr$-measurable if $J=\widehat J\circ\repr$. This covers prompts, sets, multisets, graphs, embeddings, and hand-designed statistics.
We now prove that better reasoning cannot replace lineage information, so some lineage signal must reach the judge.

\begin{theorem}[Information necessity]
\label{thm:information}
Let $\mathcal D$ be $E$-lineage-rich at $P$, and let $(P_R,P_F)$ be a matched pair in $\mathcal D$ with $(P,P_F)\in E^+$. If $\repr(P_R)=\repr(P_F)$, no $\repr$-measurable judge satisfies \RRI and positive strict \FLR on this witness. The statement for $E^-$ is symmetric.
\end{theorem}
\begin{proof}[Proof]
\RRI gives $\score_J(P_R)=\score_J(P)$. Strict \FLR gives $\score_J(P_F)>\score_J(P)$, so $\score_J(P_R)\ne\score_J(P_F)$, but $\repr(P_R)=\repr(P_F)$ forces equality for any $\repr$-measurable $J$.
\end{proof}

In plain terms, if a replica and a fresh piece of evidence look identical at the judge’s interface, no downstream reasoning can ignore one while responding to the other. If $\pi$ removes only the root label, a matched replica and its fresh counterpart share a $\pi$-image. Hence, no lineage-blind judge $J=\widehat J\circ\pi$, including an LLM judge whose prompt represents the two cases identically, satisfies \RRI and positive strict \FLR on an $E^+$ witness.

\subsection{Strategic Consequences}
\label{sec:strategic}
\begin{corollary}[Coalition replication-proofness and strict dominance]
\label{cor:repproof}
Let $J$ satisfy \RRI on $\mathcal D$, fix $\theta$ and a game $\mathcal G$ with traces in $\mathcal D$. (a) $(J,\theta)$ is coalition replication-proof for any coalition and any nonnegative costs. (b) If $c^r_i=0$ for $i\in M$, every coalition replication deviation is payoff-equivalent. Replication is then weakly, not strictly, dominated by abstaining. (c) If $c^r_i>0$ and $\alpha_i$ contains $k\ge1$ actions from $\mathcal R$, deleting them raises $U_i$ by exactly $kc^r_i$ for every $\alpha_{-i}$: $\alpha_i$ is strictly dominated, and no best response of any agent or coalition contains a replication action.
\end{corollary}
\begin{proof}[Proof]
(a) Each inserted or deleted $\mathcal R$-action is a root-replica step or its reversal, so by induction \RRI fixes the score at every prefix, whoever acts. (b) Zero replication cost leaves utilities exactly equal. (c) By (a) the verdict term is identical for every $\alpha_{-i}$, while the cost term drops by $kc^r_i$.
\end{proof}

Strategically, no agent or coalition can change an \RRI judge's verdict by repackaging represented evidence, so repackaging is strictly dominated whenever it is costly.

\begin{proposition}[Profitable deviations for occurrence-sensitive and content-keyed judges]
\label{prop:profitable}
At $\theta=\tfrac12$, there is a profile $P$, an agent $i$, and $\bar c>0$ such that (a) for Vote and occurrence \dfquad, a single relay (O2) strictly increases $U_i$ for every $c^r_i\in(0,\bar c)$; (b) for the content-keyed quotient (key $(\kappa,\lambda,p,s)$, hence \CRI), a single excerpt split (O5) does so. Neither judge is coalition replication-proof, and \RRI is strictly stronger than \CRI.
\end{proposition}

\begin{example}[Directly verifiable profitable deviation]
\label{ex:deviation}
Let $v_0=\tfrac12$ with a supporter on root $\ell_1$ and an attacker on root $\ell_2$, both strength $\tfrac15$; agent 2 holds $\ell_2$, prefers verdict $0$, $\theta=\tfrac12$. Every judge below scores $\tfrac12$ and returns $\decision=1$. Agent 2 appends one occurrence with root $\ell_2$, same target, polarity, and strength $\tfrac15$: (A) relay, or (B) excerpt split (a different canonical proposition).
\begin{center}
\footnotesize
\setlength{\tabcolsep}{3.5pt}
\begin{tabular}{lcc}
\toprule
Judge & (A) relay & (B) excerpt split \\
\midrule
Vote / occurrence \dfquad & flips to $0$ & flips to $0$ \\
Content-keyed quotient & neutral & flips to $0$ \\
\textbf{Lineage quotient} $\ql$ & neutral & neutral \\
\bottomrule
\end{tabular}
\end{center}
Both deviations add an occurrence of the unit $(\ell_2,c,-)$ with strength $\tfrac15\le\rho_P=\tfrac15$. Hence, by Def.~\ref{def:matched}, $\ql$ treats both as root replicas. The content-keyed quotient absorbs deviation (A), but deviation (B) introduces the new key $(\kappa',\ell_2,c,-)$. It therefore satisfies \CRI but violates \RRI. Under deviation (B), Agent 2's payoff rises from $0$ to $w_2-c^r_2>0$ for every $c^r_2\in(0,w_2)$, yielding Prop.~\ref{prop:profitable}.
\end{example}

\section{The Lineage Quotient}
\label{sec:repr}
The occurrence multiset retains every report. The content quotient $\qc(P)$ retains one record per $(\kappa,p,s)$ and merges all roots, while the content-keyed quotient $\qk(P)$ retains one record per $(\kappa,\lambda,p,s)$–the earlier key, kept as the informative failure case. We define $U_L(P)=\{(\lambda,p,s,\rho_P(\lambda,p,s)):(\lambda,p,s)\in\uR(P)\}$, and the \emph{lineage quotient} $\ql(P)=[U_L(P)]_{\cong}$ is its isomorphism class under root-atom bijections preserving $p,s,\rho$.

In plain terms, $\ql$ treats repeated renderings of the same root--role
as one unit, but keeps a genuinely fresh root as a new unit
(Table~\ref{tab:keys}).

\begin{table}[t]
\caption{Units retained by four representations under three interventions on a role already in the profile.}
\label{tab:keys}
\centering
\footnotesize
\setlength{\tabcolsep}{3pt}
\begin{tabular}{lccc}
\toprule
Representation (key) & Replica & Excerpt split & Fresh root \\
\midrule
Occurrence multiset & 2 & 2 & 2 \\
Content quotient $\qc$ $(\kappa,p,s)$ & 1 & 2 & 1 \\
Content-keyed $\qk$ $(\kappa,\lambda,p,s)$ & 1 & 2 & 2 \\
\textbf{Lineage quotient $\ql$ $(\lambda,p,s)$} & 1 & 1 & 2 \\
\bottomrule
\end{tabular}
\end{table}

\begin{proposition}[Content collapse is over-invariant]
\label{prop:contentcollapse}
Suppose $P$ contains role $(\kappa,p,s)$ at strength $t$, and let $b_F$ be a matched fresh-root occurrence of the same role and strength. Then $\qc(P\multisetplus\{b_F\})=\qc(P)$, so every $\qc$-measurable judge violates strict \FLR at $(P,P_F)$ directly, without invoking Thm.~\ref{thm:information}.
\end{proposition}
Content-only deduplication is thus too coarse: it can merge independent
evidence from distinct roots simply because the reports state the same
content in the same role.
The quotients $\qc$ and $\ql$ are \emph{incomparable}: neither is a function of the other, since $\qc$ discards roots while $\ql$ discards content.

Write $P\approxL P'$ when one profile reaches the other by finitely many root-replica additions, reversals (deletions that remove neither the unit nor its aggregate strength), and root bijections. A domain $\mathcal D$ is \emph{orbit-connected} if any $P,P'\in\mathcal D$ with $P\approxL P'$ are joined by a path in $\mathcal D$; without it, \RRI can be vacuous between profiles with equal quotients.

\begin{lemma}[Quotient completeness]
\label{lem:complete}
For finite well-formed profiles under Assumptions~\ref{as:root}--\ref{as:role}, $\ql(P)=\ql(P')$ iff $P\approxL P'$.
\end{lemma}
\begin{proof}[Proof]
Root-replica steps, reversals, and root renamings preserve the quotient class by absorption of $\max$. Conversely, given equal classes, rename by a witnessing bijection, reduce each unit to a strength-maximal representative (a permitted reversal), then exchange representatives by add--delete pairs, since identity, wording, and content are outside the key.
\end{proof}

The next result shows that a judge ignores replicas and root names exactly when it reads the profile only through the lineage quotient.
\begin{theorem}[Factorization]
\label{thm:factor}
On an orbit-connected $\mathcal D$, a deterministic $J:\mathcal D\to[0,1]$ satisfies \RRI and \LNI iff $J=\widetilde J\circ\ql$ for a unique $\widetilde J$ on $\ql(\mathcal D)$.
\end{theorem}
\begin{proof}[Proof]
If such $\widetilde J$ exists, replica steps and renamings preserve the quotient, hence the score. Conversely, \RRI handles additions/reversals, and \LNI handles renamings, so $J$ is invariant along every permitted step; Lemma~\ref{lem:complete} and orbit connectivity give an internal path between any two profiles with equal quotients, so $\widetilde J(\ql(P))=J(P)$ is well defined.
\end{proof}

Hence, building a judge with \RRI and \LNI reduces to choosing a score $\widetilde J$ on lineage quotients, and every such score gives a valid judge. In practice, however, a judge sees the profile through an observation interface $\repr$, such as a prompt, rather than through $\ql$ itself. Call $\repr$ \emph{invariant-complete} on orbit-connected $\mathcal D$ if every deterministic \RRI--\LNI judge factors through it.

\begin{theorem}[Coarsest invariant-complete representation]
\label{thm:minimal}
On a fixed orbit-connected $\mathcal D$, if $\repr$ is invariant-complete then $\ql|_{\mathcal D}=g\circ\repr$ for some $g$. Since $\ql$ is itself invariant-complete (Thm.~\ref{thm:factor}), it is, up to recoding, the coarsest such representation, and this holds without excluding source-reliability weighting, which enters through represented attributes that $\ql$ retains.
\end{theorem}
\begin{proof}[Proof]
If $\repr(P)=\repr(P')$ but $\ql(P)\ne\ql(P')$, the indicator $J_P(R)=\mathbf 1[\ql(R)=\ql(P)]$ satisfies \RRI and \LNI, so invariant-completeness makes it $\repr$-measurable---contradicting $\repr(P)=\repr(P')$.
\end{proof}

Together with Thm.~\ref{thm:factor}, this makes $\ql$ both sufficient and minimal: any representation supporting \emph{every} deterministic \RRI--\LNI judge must allow $\ql$ to be recovered. A single judge may need less; on direct stars, \lqdf needs only $v_0$ and both aggregates, a statistic strictly coarser than $\ql$ (supplement).

% This bound is relative to the \emph{full} deterministic \RRI--\LNI class on a fixed domain: a restricted family can need strictly less---for instance, on the audited direct-star grid, \dfquad depends only on $(v_0,$ signed support/attack aggregates$)$, a statistic strictly coarser than $\ql$ that still satisfies all principles for that one judge (supplement).

\section{Lifting the Lineage Quotient to QBAFs}
\label{sec:qbaf}
An occurrence QBAF is $Q^o=\langle C\dot\cup O,R^-,R^+,\tau,c,\equivL\rangle$ with structural claim nodes $C\ni c$, report occurrences $O$, and $\equivL$ defined only on $O$. Thus, claim nodes are never merged by lineage. Let $q_0$ map each occurrence to its root--role class. 
The graph is \emph{quotient-compatible} if equivalent occurrences have identical signed incoming and outgoing neighborhoods after all endpoints are mapped by $q_0$. That is, occurrences merged into one class must have the same signed incoming and outgoing neighbors at the quotient level. Canonical content may still differ within a class, so different excerpts of the same root can remain in that class.

For each class, define its strength by $\bar\tau([a])=\max\{\tau(b):b\equivL a\}$, which is independent of the chosen representative by construction. The quotient $Q^o/{\equivL}$ collapses each root--role class into a single node and removes duplicate induced edges. Classes that violate quotient compatibility are treated as malformed.
A QBAF is \emph{lineage-stratified} if there exists a rank function $r$ such that $r(a)<r(b)$ for every edge $(a,b)$ and $r(a)=r(b)$ whenever $a\equivL b$. The quotient of every lineage-stratified QBAF is therefore acyclic. The proof is given in the supplement.
For a deterministic, isomorphism-invariant semantics $\sigma$, its lineage-quotient lifting is $\operatorname{LQ}[\sigma](Q^o)=\sigma(Q^o/{\equivL})$, defined on $\sigma$-admissible $Q^o$.

\begin{theorem}[Generic lifting]
\label{thm:lifting}
For every $\sigma$-admissible $Q^o$: (1) $\operatorname{LQ}[\sigma]$ satisfies structural \RRI (invariance to a new occurrence node duplicating a class and its quotient neighborhood, at strength at most the class aggregate); (2) it satisfies \LNI; (3) with singleton classes it equals $\sigma$ up to isomorphism.
\end{theorem}
\begin{proof}[Proof]
A structural replica adds no quotient node or edge, since its class and neighborhood already exist, and $\max$ absorbs its strength; a root renaming only renames classes; singleton classes reproduce the raw graph. Isomorphism invariance of $\sigma$ transfers each identity.
\end{proof}

Operationally, the lifting is \emph{quotient first, evaluate second}: replicas disappear before $\sigma$ runs, root names cannot matter to an isomorphism-invariant $\sigma$, and graphs without replicas are scored as before.

Call $\sigma$ \emph{signed-aggregative} if $\sigma(c)=I(v_0,A^+(V^+),A^-(V^-))$ with $A^\pm$ permutation-invariant and $I$ strictly monotone in each aggregate.

\begin{theorem}[Aggregative \FLR characterisation]
\label{thm:characterization}
Let every eligible fresh input be a direct topic leaf. $\operatorname{LQ}[\sigma]$ satisfies positive \FLR iff $A^+(V^+\multisetplus\{t\})>A^+(V^+)$ for every eligible pair; negative \FLR symmetrically for $A^-$.
\end{theorem}
\begin{proof}[Proof]
A fresh root creates one new singleton-strength leaf, changing only $V^+$ (resp. $V^-$); strict monotonicity of $I$ makes the score move iff the aggregate does.
\end{proof}

Intuitively, the quotient delivers a fresh direct report as a new input, but responding is the aggregator's job: the score moves its way exactly when the pooled value grows, which a maximum aggregator fails to do for inputs not above the current maximum.

\dfquad \cite{rago2016dfquad} aggregates by noisy-OR, $F(v_1,\ldots,v_k)=1-\prod_i(1-v_i)$, and scores the topic by
\begin{equation}
\sigma(c)=
\begin{cases}
 v_0-v_0(v_a-v_s), & v_a>v_s,\\
 v_0, & v_a=v_s,\\
 v_0+(1-v_0)(v_s-v_a), & v_s>v_a,
\end{cases}
\label{eq:dfquad}
\end{equation}
where $v_s,v_a$ are the supporter/attacker aggregates and $v_0=\tau(c)$.

\begin{proposition}[Strict direct response]
\label{prop:response}
Let $v_0\in(0,1)$. If $v_s<1$, a fresh direct supporter of strength $t\in(0,1]$ strictly increases the \dfquad score with $v_a$ fixed; if $v_a<1$, the symmetric attacker strictly decreases it.
\end{proposition}
Noisy-OR gives $v'_s=v_s+t(1-v_s)>v_s$, and Eq.~\eqref{eq:dfquad} is strictly increasing in $v_s$ across all three branches (full case analysis in the supplement).

\begin{corollary}[The principles are jointly satisfiable]
\label{cor:consistent}
On well-formed lineage-stratified acyclic occurrence QBAFs, \lqdf satisfies structural \RRI, \LNI, and both signed strict \FLR forms, and equals \dfquad when every class is a singleton.
\end{corollary}

\subsection{Propagation and a Quantitative Modulus}
\label{sec:propagation}

A fresh report may attach to an intermediate claim rather than the topic; whether its effect reaches the topic depends on the strengths and edge polarities along the way. Say $d$ \emph{influences $n$ sign-consistently} with sign $\varepsilon(n)$ if at least one directed path leads from $d$ to $n$ (the trivial path when $n=d$) and either all such paths contain an even number of attack edges, $\varepsilon(n)=+1$, or all contain an odd number, $\varepsilon(n)=-1$.

% A fresh report may attach to an intermediate claim rather than the topic; whether its effect reaches the topic depends on edge polarities along the way. Say $d$ \emph{influences $n$ sign-consistently} with sign $\varepsilon(n)$ if every directed path from $d$ to $n$ has the same attack-edge parity $\varepsilon(n)$.

\begin{theorem}[Strict propagation along sign-consistent paths]
\label{thm:propagate}
Let $Q=Q^o/{\equivL}$ be acyclic with every base strength in $(0,1)$. Add one eligible fresh root as a direct leaf of $d$, polarity $\varepsilon_0$, strength $t\in(0,1)$, and suppose $d$ influences $c$ sign-consistently with sign $\varepsilon$. Then $\operatorname{LQ}[\dfquad](c)$ changes strictly, with sign $\varepsilon_0\varepsilon$; this contains the direct-leaf case $d=c$, $\varepsilon=+1$.
\end{theorem}
% \begin{proof}[Proof]
% Under the strength hypothesis every evaluated aggregate stays in $(0,1)$, so noisy-OR is strictly monotone throughout and the \dfquad influence function $\phi(\tau,v_s-v_a)$ is piecewise linear with strictly positive slopes. Induct over rank on the nodes reachable from $d$ toward $c$: sign consistency propagates backward, so a node's changing supporters and attackers move its $(v_s-v_a)$ difference with matching, reinforcing sign rather than competing, and strict monotonicity of $\phi$ transfers this to the score at every step, ending at $c$ (full induction in the supplement).
% \end{proof}

% Sign consistency cannot be dropped: on a diamond where an intermediate root supports two claims that respectively support and attack the topic with symmetric base strengths $\tfrac12$, a fresh supporter of the shared intermediate claim moves both intermediate scores strictly yet leaves the topic score exactly unchanged (worked example in the supplement), because the two opposite-parity paths cancel exactly.

\begin{proof}[Proof]
Interior strengths keep every aggregate in $(0,1)$, where noisy-OR and the \dfquad influence function are strictly monotone; sign consistency makes all changes at a node reinforce rather than compete, so induction over rank carries a strict change from $d$ to $c$ (supplement).
\end{proof}

In other words, with interior strengths, fresh evidence deeper in the graph still moves the topic strictly when all paths agree on the net sign; each attack edge flips the sign, so supporting an attacker of the topic lowers it. The condition cannot be dropped: on a symmetric diamond with opposite-parity paths, a fresh supporter moves both intermediate claims yet leaves the topic exactly unchanged, as the two paths cancel (supplement).

\begin{theorem}[Tight response modulus]
\label{thm:modulus}
Let $v_0\in(0,1)$, $v_s,v_a\in[0,1)$. A fresh direct supporter of strength $t\in(0,1]$ changes the topic score by $\Delta\sigma(c)\ge t(1-v_s)\min\{v_0,1-v_0\}>0$; a fresh attacker decreases it by at least $t(1-v_a)\min\{v_0,1-v_0\}$. Both bounds are attained, so \lqdf satisfies \QFLR with modulus $\delta(t,v,v_0)=t(1-v)\min\{v_0,1-v_0\}$.
\end{theorem}
\begin{proof}[Proof]
Write $\sigma(c)=\phi(d)$, $d=v_s-v_a$: $\phi$ is continuous, piecewise linear with slopes $v_0,1-v_0\ge m=\min\{v_0,1-v_0\}$. The supporter shifts $d\mapsto d+\delta$, $\delta=t(1-v_s)$; any such $\phi$ satisfies $\phi(d+\delta)-\phi(d)\ge m\delta$ in every branch, including the crossing case, with equality attained at $v_0\ge\tfrac12$, $d\ge0$ (full derivation in the supplement).
\end{proof}

In practical terms, the modulus is the least effect a fresh direct report is guaranteed to have; it grows with the strength $t$ and shrinks as the same-side aggregate $v$ saturates or $v_0$ nears $0$ or $1$. The bound is tight, and its factor $1-v$ stems from noisy-OR, not from quotienting.

\paragraph{Operation separation.} Every finite interleaving of relay, paraphrase, alias resubmission, or excerpt split by any coalition is a structural replica extension and leaves $\operatorname{LQ}[\sigma]$ and every threshold decision unchanged (Thm.~\ref{thm:lifting}(1)); every eligible fresh contribution strictly moves the score in its sign, directly (Thm.~\ref{thm:characterization}) or through a sign-consistent path (Thm.~\ref{thm:propagate}). This is the graph-level form of Cor.~\ref{cor:repproof}, with the scope limits of Sec.~\ref{sec:discussion}.

\input{temp/07_Exp_New}

\section{Discussion}
\label{sec:discussion}
 
\textit{(1) Supplied Information and Dependence Modeling.}
The lineage partition is an input, not an inferred output. Thm.~\ref{thm:information} requires only that \emph{some} interface signal separate a replica from its fresh counterpart, and building and authenticating that signal, including deciding how documents decompose into roots, is upstream work (Asm.~\ref{as:root}). A partition also cannot express graded dependence, which needs probabilistic models as in the pooling and jury-theorem literature (Sec.~\ref{sec:related}). In the direct same-sign case, false splits inflate aggregation and false merges suppress it, and quantifying this under realistic error rates remains open.
\textit{(2) Content Representation.}
Keying by root and role blocks excerpt splitting but removes canonical content from aggregation, so roots differ only through their represented strengths. Content-sensitive weighting must enter as an explicit attribute, like source reliability.
\textit{(3) Strategic and Graph Scope.}
Cor.~\ref{cor:repproof} gives replication resistance under authenticated roots, not dominant-strategy truthfulness or false-name-proofness, and it excludes forged identities, strategic misreports, strength revelation, and genuinely new evidence. Extending propagation (Thm.~\ref{thm:propagate}) beyond \dfquad{} remains open.

\section{Conclusion}
\label{sec:conclusion}
Visible agreement can reflect replication or corroboration, and a judge handles it correctly only if its representation separates further renderings of a represented root from distinct roots. Under \RRI, replication cannot change the verdict in the judgment game and is strictly dominated whenever it is costly, whereas occurrence-sensitive and content-keyed judges admit profitable deviations (Ex.~\ref{ex:deviation}). By Thm.~\ref{thm:information}, no interface that represents a replica and a fresh root identically can be both replica-neutral and strictly fresh-responsive. The lineage quotient is the coarsest representation sufficient for every deterministic \RRI--\LNI judge on an orbit-connected domain. Lifted to QBAFs, a fresh root's effect propagates strictly along sign-consistent paths with the tight \dfquad{} modulus $t(1-v)\min\{v_0,1-v_0\}$. Exhaustive and AVeriTeC-derived audits confirm the predicted separation under a supplied partition.
Future work will use graded, estimated dependence, measure sensitivity to false splits and merges, test the information boundary on prompted LLM judges, and audit excerpt splits and intermediate-node interventions at scale.

\bibliographystyle{ACM-Reference-Format}
\bibliography{references}

\appendix
\section*{Appendices}
This supplement contains only material that is omitted or abbreviated in the main paper. It gives the deferred technical proofs for the lineage quotient and its QBAF lifting, works out the counterexamples and boundary cases used by the main claims, and records the experimental details needed to reproduce the structural and AVeriTeC audits. Definitions, theorem statements, and empirical conclusions already given in the main paper are not repeated unless they are needed locally for a proof.

Numbering such as ``Thm.~2'' or ``Sec.~7.2'' refers to the main paper. We use the same notation throughout. The supplement has three roles. Sec.~\ref{sec:profile-proofs} fills the proof steps omitted from the profile-level representation results. Sec.~\ref{sec:qbaf-proofs} gives the deferred QBAF and DF-QuAD derivations, including the cancellation witness required to delimit Thm.~6. Sec.~\ref{sec:empirical-details} gives only the experimental construction, preprocessing, uncertainty estimation, and sensitivity analyses, which are not fully detailed in the main paper.

\section{Proof Details}
\label{sec:profile-proofs}

\subsection{Why Absorption Fixes the Root Aggregate}
The main paper uses an absorptive within-root aggregate and states that this uniquely gives the maximum. The following proposition supplies the omitted argument.

\begin{proposition}[Absorption forces the maximum]
\label{sprop:absorption}
Let \(\rootagg\) be defined on finite nonempty multisets of \([0,1]\). Assume \(\rootagg\{t\}=t\) and
\[
 t\le \rootagg V
 \quad\Longrightarrow\quad
 \rootagg(V\multisetplus\{t\})=\rootagg V.
\]
Then \(\rootagg V=\max V\) for every finite nonempty \(V\). Conversely, \(\max\) satisfies the singleton condition and absorption.
\end{proposition}

\begin{proof}
Induct on \(|V|\). The singleton case is immediate. Suppose the claim holds for multisets of size \(n\), and let \(|V|=n+1\). Choose one minimum element \(t\in V\) and write \(W=V\setminus\{t\}\). By the induction hypothesis, \(\rootagg W=\max W\). Since \(t\le\max W=\rootagg W\), absorption gives
\[
 \rootagg V
 =\rootagg(W\multisetplus\{t\})
 =\rootagg W
 =\max W
 =\max V.
\]
The converse follows directly from the definition of \(\max\).
\end{proof}

The singleton and absorption conditions are both needed. For example, the arithmetic mean is symmetric and idempotent on equal inputs, but appending a weaker rendering can lower the aggregate. Such a rule would let replication change an already represented root. Proposition~\ref{sprop:absorption} excludes this effect and justifies the use of \(\rho_P(\upsilon)=\max\{\tau(a):\uR(a)=\upsilon\}\).

\subsection{Quotient Completeness}
Under the proof below, a \emph{reversal} deletes an occurrence only when the remaining profile still contains the same root--role unit with the same root strength. Thus, a permitted reversal is exactly the inverse of a root-replica addition.

\begin{proof}[Full proof of Lemma~1 in the main paper]
First assume \(P\approxL P'\). A root-replica addition leaves the represented units and their maximum strengths unchanged by Proposition~\ref{sprop:absorption}. A permitted reversal has the same property by definition. A consistent bijection of root labels changes only the names of root atoms. Hence every permitted step preserves the isomorphism class \(\ql\), and therefore \(\ql(P)=\ql(P')\).

Conversely, assume \(\ql(P)=\ql(P')\). Choose a root-label bijection witnessing equality of the two quotient classes and apply it to one profile, so the weighted root--role records agree exactly. In each profile, consider one root--role unit at a time. Retain one occurrence of maximal strength and delete the other occurrences. Every deletion is a permitted reversal because a maximal-strength representative remains, so neither the unit nor its aggregate strength changes. This reduces the two profiles to \(C\) and \(C'\), each with one representative per root--role unit.

The representatives in \(C\) and \(C'\) have the same root, target, sign, and maximal strength, but may differ in displayed identity, wording, or canonical proposition. These fields are outside the quotient key. Replace each representative of \(C\) by its counterpart in \(C'\) using two permitted steps: first append the counterpart, which is a root replica at the existing aggregate strength, then delete the old representative. Repeating this finite operation yields \(C'\). Reversing the initial reductions restores the multiplicities of \(P'\). Hence \(P\approxL P'\).
\end{proof}

The use of a maximal-strength representative is essential. Deleting the only occurrence that attains the maximum could lower \(\rho_P\), so that deletion would not be a permitted reversal.

\subsection{Factorization and Minimal Sufficiency}

\begin{proof}[Full proof of Thm.~2 in the main paper]
Suppose first that \(J=\widetilde J\circ\ql\). A root-replica addition or reversal leaves \(\ql\) unchanged, and a bijective root renaming changes only the representative of the same quotient class. Therefore \(J\) satisfies \RRI and \LNI.

For the converse, suppose \(J\) satisfies \RRI and \LNI on an orbit-connected domain \(\mathcal D\). Let \(P,P'\in\mathcal D\) with \(\ql(P)=\ql(P')\). Lemma~1 gives a finite sequence of replica additions, permitted reversals, and root renamings from \(P\) to \(P'\). Orbit connectivity gives such a sequence inside \(\mathcal D\). \RRI makes every addition neutral. If a step deletes \(b\), then \(b\) is a root replica relative to the smaller endpoint, so \RRI at that endpoint also makes the reversal neutral. \LNI makes every root renaming neutral. Thus \(J(P)=J(P')\).

Define \(\widetilde J(\ql(P))=J(P)\). The preceding argument makes this definition independent of the chosen representative. It is unique because every point of \(\ql(\mathcal D)\) has a representative in \(\mathcal D\). Hence \(J=\widetilde J\circ\ql\).
\end{proof}

\begin{proof}[Full proof of Thm.~3 in the main paper]
Let \(\repr\) be invariant-complete on the fixed orbit-connected domain \(\mathcal D\). We first show that
\[
 \repr(P)=\repr(P')
 \quad\Longrightarrow\quad
 \ql(P)=\ql(P').
\]
Assume otherwise. Fix \(P\) and define the quotient-class indicator
\[
 J_P(R)=\mathbf 1[\ql(R)=\ql(P)].
\]
Every root-replica step and every root renaming preserves \(\ql\), so \(J_P\) satisfies \RRI and \LNI. Invariant-completeness therefore makes \(J_P\) measurable with respect to \(\repr\). But if \(\repr(P)=\repr(P')\) while \(\ql(P)\ne\ql(P')\), measurability would require \(J_P(P)=J_P(P')\), whereas the two values are \(1\) and \(0\). This is a contradiction.

We can therefore define \(g\) on \(\repr(\mathcal D)\) by \(g(\repr(P))=\ql(P)\). The implication above makes \(g\) well defined, and \(\ql|_{\mathcal D}=g\circ\repr\). Thm.~2 shows that \(\ql\) is itself invariant-complete, so no strictly coarser representation can support the full deterministic \RRI--\LNI judge class on the same domain.
\end{proof}

\paragraph{Why a particular judge can retain less.}
The minimality statement is relative to the full deterministic \RRI--\LNI class. On a direct-star QBAF, \dfquad depends only on
\[
 \repr_{\mathrm{agg}}(P)=
 \left(
 v_0,
 1-\!\!\prod_{s(\upsilon)=+}(1-\rho_P(\upsilon)),
 1-\!\!\prod_{s(\upsilon)=-}(1-\rho_P(\upsilon))
 \right).
\]
This statistic is strictly coarser than \(\ql\). For example, supporter multisets \(\{0.2,0.2,0.8\}\) and \(\{0.2,0.6,0.6\}\) have different lineage quotients but the same noisy-OR residual, \(0.8\cdot0.8\cdot0.2=0.8\cdot0.4\cdot0.4=0.128\), and therefore the same support aggregate \(0.872\). The statistic is sufficient for this one judge, but it cannot recover arbitrary quotient-class indicators. This is the restricted-family qualification referenced after Thm.~3 in the main paper.

\subsection{Admissibility of Replication-only Deviations}
The proof of Corollary~1 uses the fact that inserting or deleting replication actions does not invalidate the remaining trace. This follows from the action definitions.

\begin{lemma}[Replication actions are admissibility-neutral]
Let \(\alpha_i'\) be a replication-only deviation from \(\alpha_i\). Every non-replication action admissible in the original trace remains admissible in the deviated trace, and conversely.
\end{lemma}

\begin{proof}
An action in \(\mathcal R=\{\mathrm{O2},\ldots,\mathrm{O5}\}\) can be taken only after its root--role unit is already represented, and its strength is at most the current root strength. It therefore creates no new root, creates no new root--role unit, and does not change any \(\rho_P\). The admissibility conditions of the non-replication actions depend only on the represented roots, represented units, current root strengths, and the acting agent's root endowment. Inserting or deleting actions from \(\mathcal R\) changes none of these quantities. Hence the same non-replication actions remain admissible in the same relative order.
\end{proof}

\section{Deferred QBAF and DF-QuAD Derivations}
\label{sec:qbaf-proofs}

\subsection{Acyclicity after Lineage Quotienting}
The main paper states that lineage stratification is sufficient for the quotient graph to remain acyclic.

\begin{lemma}[Acyclicity under lineage stratification]
If an occurrence QBAF is lineage-stratified, its lineage quotient is acyclic.
\end{lemma}

\begin{proof}
Let \(r\) be the rank function from the definition of lineage stratification. Equivalent occurrence nodes have the same rank, so each quotient class \([a]_R\) has a well-defined rank \(\bar r([a]_R)=r(a)\). Every quotient edge is induced by a raw edge, and every raw edge strictly increases \(r\). Hence every quotient edge strictly increases \(\bar r\). A directed quotient cycle would imply
\[
 \bar r(z_1)<\bar r(z_2)<\cdots<\bar r(z_k)<\bar r(z_1),
\]
which is impossible.
\end{proof}

Lineage stratification is sufficient but not necessary. An acyclic raw graph can fail this condition and can even acquire a quotient cycle. For example, if \(a_1\to b\to a_2\) with \(a_1\equivL a_2\), collapsing \(a_1\) and \(a_2\) creates a directed cycle in the quotient.

\subsection{Full Case Analysis for Proposition~4}
Write \(d=v_s-v_a\). The DF-QuAD topic score can be written as:
\[
 \phi(d)=
 \begin{cases}
 v_0(1+d), & d\le 0,\\
 v_0+(1-v_0)d, & d\ge 0.
 \end{cases}
\]
The two branches meet at \(d=0\), and their slopes are \(v_0\) and \(1-v_0\).

\begin{proof}[Full proof of Proposition~4 in the main paper]
Add a direct supporter of strength \(t\in(0,1]\) when \(v_s<1\). Noisy-OR changes the support aggregate by
\[
 \delta=v_s'-v_s=t(1-v_s)>0.
\]
If \(v_s'\le v_a\), both scores are on the attack-dominant branch and
\[
 \phi(d+\delta)-\phi(d)=v_0\delta>0.
\]
If \(v_a\le v_s\), both are on the support-dominant branch and
\[
 \phi(d+\delta)-\phi(d)=(1-v_0)\delta>0.
\]
If \(v_s<v_a<v_s'\), the intervention crosses the branch point. Then
\[
 \phi(d+\delta)-\phi(d)
 =(1-v_0)(v_s'-v_a)+v_0(v_a-v_s)>0.
\]
The boundary cases set one of the two terms to zero. Thus, every eligible direct supporter strictly increases the score. The attacker statement is symmetric because increasing \(v_a\) shifts \(d\) in the negative direction, and both branch slopes with respect to \(v_a\) are strictly negative.
\end{proof}

\subsection{Full Induction for Theorem 6}
Throughout this subsection, \(Q\) is an acyclic quotient QBAF evaluated by \dfquad. For a node \(n\), let \(v_s(n)\) and \(v_a(n)\) be its noisy-OR support and attack aggregates. Then
\[
 \sigma(n)=\phi\!\left(\tau(n),v_s(n)-v_a(n)\right),
\]
where:
\[
 \phi(\tau,\Delta)=
 \begin{cases}
 \tau(1+\Delta), & \Delta\le 0,\\
 \tau+(1-\tau)\Delta, & \Delta\ge 0.
 \end{cases}
\]

\begin{lemma}[Non-degeneracy]
\label{slem:nondegenerate}
If every node has base strength in \((0,1)\), then every evaluated node score also lies in \((0,1)\).
\end{lemma}

\begin{proof}
Induct over a topological order. A leaf has score \(\tau(n)\in(0,1)\). For an interior node, the induction hypothesis gives all parent scores in \((0,1)\), so both noisy-OR aggregates lie in \([0,1)\). If attacks dominate, \(\sigma(n)=\tau(n)(1+\Delta)\in(0,\tau(n)]\). If supports dominate, \(\sigma(n)=\tau(n)+(1-\tau(n))\Delta\in[\tau(n),1)\). Thus \(\sigma(n)\in(0,1)\) in either case.
\end{proof}

A useful consequence is the strict monotonicity of noisy-OR on the evaluated domain. If one or more inputs strictly increase and none decrease, with every input below one, then the aggregate strictly increases. The symmetric statement holds for decreases.

\begin{lemma}[Parity is well defined on the relevant subgraph]
\label{slem:parity}
Suppose every directed path from \(d\) to \(c\) has the same attack-edge parity. Let \(R^*\) be the set of nodes lying on at least one directed path from \(d\) to \(c\). Then each \(n\in R^*\) has a well-defined parity \(\varepsilon(n)\) from \(d\), with \(\varepsilon(d)=+1\). If \((m,n)\) is a support edge, \(\varepsilon(n)=\varepsilon(m)\). If it is an attack edge, \(\varepsilon(n)=-\varepsilon(m)\). Moreover, every in-neighbor of \(n\in R^*\) that is reachable from \(d\) also belongs to \(R^*\).
\end{lemma}

\begin{proof}
If two paths from \(d\) to some \(n\in R^*\) had opposite parity, append a common suffix from \(n\) to \(c\). This would produce two \(d\)-to-\(c\) paths of opposite parity, contradicting the hypothesis. Hence \(\varepsilon(n)\) is well defined. Appending one support edge preserves parity, and appending one attack edge reverses it. For the final claim, concatenate a path from \(d\) to the in-neighbor, its edge to \(n\), and a path from \(n\) to \(c\).
\end{proof}

\begin{proof}[Full proof of Thm.~6 in the main paper]
Add the fresh leaf to \(d\) with polarity \(\varepsilon_0\in\{+1,-1\}\). Let \(R\) be the set of nodes reachable from \(d\). Nodes outside \(R\) are unchanged: by topological induction, none has a change in neighbors, so its aggregates and score remain fixed.

Now induct over a topological order of \(R^*\). We prove that every \(n\in R^*\) changes strictly with sign \(\varepsilon_0\varepsilon(n)\). At \(n=d\), Lemma~\ref{slem:nondegenerate} keeps the relevant pre-intervention aggregate below one. Proposition~4 therefore gives a strict change with sign \(\varepsilon_0\), which equals \(\varepsilon_0\varepsilon(d)\).

Let \(n\ne d\) and assume the claim for all earlier nodes in \(R^*\). Every changed in-neighbor of \(n\) lies in \(R^*\) by Lemma~\ref{slem:parity}, and at least one changed in-neighbor exists because \(n\) lies on a positive-length path from \(d\). A changed supporter \(m\) has \(\varepsilon(m)=\varepsilon(n)\), so its score changes with sign \(\varepsilon_0\varepsilon(n)\). A changed attacker has \(\varepsilon(m)=-\varepsilon(n)\), so its score changes with the opposite sign. Strict monotonicity of noisy-OR therefore makes \(v_s(n)\) move in sign \(\varepsilon_0\varepsilon(n)\) whenever a support input changes, while \(v_a(n)\) moves in the opposite sign whenever an attack input changes. The two channels consequently reinforce the difference \(v_s(n)-v_a(n)\). This difference changes strictly with sign \(\varepsilon_0\varepsilon(n)\).

By Lemma~\ref{slem:nondegenerate}, both slopes of \(\phi(\tau(n),\cdot)\) are strictly positive. The score \(\sigma(n)\) therefore changes with the same strict sign. Taking \(n=c\) gives the sign \(\varepsilon_0\varepsilon(c)\), which is exactly the sign stated in Thm.~6.
\end{proof}

\subsection{Why Sign Consistency cannot be Removed}
The main paper states that opposite-parity paths can cancel exactly. The following worked example gives the promised witness.

\begin{example}[Exact cancellation on a diamond]
Let \(c,m_1,m_2,d\) all have base strength \(\tfrac12\). Add edges
\[
 d\xrightarrow{+}m_1,\qquad
 d\xrightarrow{+}m_2,\qquad
 m_1\xrightarrow{+}c,\qquad
 m_2\xrightarrow{-}c,
\]
and let one report \(r_1\) of strength \(\tfrac12\) support \(d\). Initially,
\[
 \sigma(d)=\tfrac34,
 \qquad
 \sigma(m_1)=\sigma(m_2)=\tfrac78,
 \qquad
 \sigma(c)=\tfrac12.
\]
Now add a fresh supporter \(r_2\) of \(d\), also of strength \(\tfrac12\). The support aggregate at \(d\) becomes \(\tfrac34\), so
\[
 \sigma(d)=\tfrac78,
 \qquad
 \sigma(m_1)=\sigma(m_2)=\tfrac{15}{16}.
\]
Both intermediate scores increase strictly, but the topic still receives equal support and attack aggregates. Hence \(\sigma(c)=\tfrac12\) exactly as before. The two paths from \(d\) to \(c\) have opposite parity, so the zero topic response is caused by cancellation rather than saturation.
\end{example}

\subsection{Full Derivation of the Tight Response Modulus}

\begin{proof}[Full proof of Thm.~7 in the main paper]
Consider a fresh direct supporter of strength \(t\). Let
\[
 m=\min\{v_0,1-v_0\},
 \qquad
 \delta=t(1-v_s)>0,
 \qquad
 d=v_s-v_a.
\]
The intervention maps \(d\) to \(d+\delta\). If \(d+\delta\le0\), the entire change lies on the attack-dominant branch and
\[
 \Delta\sigma=v_0\delta\ge m\delta.
\]
If \(d\ge0\), the entire change lies on the support-dominant branch and
\[
 \Delta\sigma=(1-v_0)\delta\ge m\delta.
\]
If \(d<0<d+\delta\), split the displacement at zero:
\[
 \Delta\sigma
 =v_0(-d)+(1-v_0)(d+\delta).
\]
The two nonnegative displacement lengths sum to \(\delta\), and each is multiplied by a slope at least \(m\). Hence \(\Delta\sigma\ge m\delta\). Therefore
\[
 \Delta\sigma
 \ge t(1-v_s)\min\{v_0,1-v_0\}.
\]
The attacker case follows by the symmetric shift \(d\mapsto d-t(1-v_a)\).

The bound is tight. If \(v_0\ge\tfrac12\), then \(m=1-v_0\), and any support-dominant configuration that remains on that branch attains equality. If \(v_0\le\tfrac12\), an attack-dominant configuration that remains on the attack branch attains equality. The attacker's construction is symmetric.
\end{proof}

\paragraph{Exact minimum on the direct-star grid.}
For a selected supporter of strength \(t\) in a supporter multiset \(V\),
\[
 t(1-v_s)
 =t(1-t)\!\prod_{v\in V\setminus\{t\}}(1-v).
\]
On \(G=\{0.2,0.4,0.6,0.8\}\), the first factor is at least \(0.16\). With at most two other supporters, the remaining product is at least \(0.2^2=0.04\). Thus \(t(1-v_s)\ge0.0064\). Also \(\min\{v_0,1-v_0\}\ge0.2\), so the modulus is at least \(0.00128\). Equality is realised, for example, at \(v_0=0.8\), supporter strengths \((0.2,0.8,0.8)\), attacker strength \(0.2\), and intervention strength \(t=0.2\). This proves the analytic minimum reported next to Table~2 in the main paper.

\subsection{Two Additional Boundary Calculations}

\paragraph{Effect of exact split and merge errors.}
Suppose \(m>1\) direct same-sign units all have strength \(t\in(0,1)\). If one true root is falsely represented as \(m\) roots, noisy-OR changes its aggregate contribution from \(t\) to
\[
 1-(1-t)^m>t.
\]
A false split therefore inflates the same-sign aggregate. Merging \(m\) truly distinct equal-strength roots into one reverses the change and suppresses the aggregate. This is the direct same-sign calculation referenced in Sec.~8 of the main paper. It does not quantify errors in mixed-sign or multi-layer graphs.

\paragraph{Computational cost after root keys are available.}
With expected constant-time hashing, quotient construction takes expected
\[
 O(|O|+|R^-|+|R^+|)
\]
time and linear space. Each occurrence is hashed once to its root--role class while maintaining a running maximum, and each signed raw edge is mapped once to quotient endpoints with duplicate edges removed. Topological \dfquad evaluation is linear in the quotient size. Testing sign consistency is also linear in the subgraph reachable from the intervention node, using a two-coloring by attack parity. These bounds exclude provenance extraction, root discovery, and root authentication, which are outside the model.

\section{Additional Experimental Specification}
\label{sec:empirical-details}
The main paper reports the audit outcomes. This section records only construction and analysis details that are needed for exact reproduction or that are not shown in the main tables.

\subsection{Structural Audits}

\paragraph{Audit implementation and the theoretical quotient.}
In the direct-star and AVeriTeC scripts, the implementation labeled ``lineage quotient'' retains one representative per \((\kappa,\lambda,s)\) key before \dfquad. On every profile evaluated in these two audits, this induces the same classes as the theoretical root--role quotient \(\ql\): all reports target the same topic, every clean report has a distinct root, and both matched interventions keep \(\kappa\) fixed. The equivalence is therefore exact on the evaluated intervention set. The audits do not test excerpt splitting, where \(\kappa\) changes while the root and role remain fixed. That separation is established analytically by Example~1 and Proposition~2 in the main paper.

\paragraph{Direct-star enumeration.}
Let \(G=\{0.2,0.4,0.6,0.8\}\). The topic-based strength and every leaf strength range over \(G\), with one to three supporters and one to three attackers. The number of clean profiles is
\[
 |G|
 \left(\sum_{n_s=1}^{3}|G|^{n_s}\right)
 \left(\sum_{n_a=1}^{3}|G|^{n_a}\right)
 =4(4+16+64)^2
 =28{,}224.
\]
The number of report positions is
\[
 4\sum_{n_s=1}^{3}\sum_{n_a=1}^{3}
 (n_s+n_a)4^{n_s+n_a}
 =153{,}216.
\]
For a selected report position, the eightfold-replica intervention adds seven copies of that report, so the profile contains eight occurrences of the selected report. The matched fresh intervention adds one report with the same content, role, and strength but a new controlled root. No sampling is used in this audit.

\paragraph{Multi-layer enumeration}
Three typed-QBAF families (main paper, Sec.~7.2) are fully enumerated. A chain has one report leaf, one intermediate claim, and the topic. Its three node strengths range over \(\{0.2,0.5,0.8\}\), and its two edges have independent signs, giving \(3^3 2^2=108\) graphs. A fork has two report leaves feeding one intermediate claim and then the topic, giving \(3^4 2^3=648\) graphs. A diamond has one report leaf feeding two intermediate claims that both feed the topic, giving \(3^4 2^4=1{,}296\) graphs. Hence, the audit contains \(2{,}052\) graphs in total.

Each report leaf is replicated once while preserving its complete signed neighborhood. This gives \(108+2(648)+1{,}296=2{,}700\) replicated positions. Each graph also receives one fresh direct supporter and one fresh direct attacker of strength \(0.4\), giving \(2(2{,}052)=4{,}104\) fresh interventions. Intermediate-node fresh-root interventions are not part of this empirical audit.
\begin{table*}[h]
\caption{AVeriTeC audit with claim-clustered 95\% bootstrap intervals. MinG is the observed finite-position minimum, and \(n_F\) is the judge-specific FASR denominator.}
\label{tab:real-full}
\centering
\resizebox{\textwidth}{!}{%
\begin{tabular}{lccccccc}
\toprule
Judge & CD & RSS$_8$ & RFR$_8$ & FLG & MinG & FASR & $n_F$\\
\midrule
Vote & .01867 [.01408,.02373] & .05724 [.04294,.07196] & .07032 [.05292,.08873] & .02194 [.01648,.02766] & .00000 & 1.000 [1.000,1.000] & 50\\
Occurrence \dfquad & .00000 [.00000,.00000] & .01137 [.00894,.01399] & .06892 [.05114,.08733] & .00880 [.00710,.01069] & $1.13\!\times\!10^{-8}$ & .980 [.933,1.000] & 50\\
Max-QuAD & .02785 [.02546,.03035] & .00000 [.00000,.00000] & .00000 [.00000,.00000] & .00000 [.00000,.00000] & .00000 & .000 [.000,.000] & 61\\
Content quotient & .00028 [.00000,.00075] & .00000 [.00000,.00000] & .00000 [.00000,.00000] & .00000 [.00000,.00000] & .00000 & .000 [.000,.000] & 50\\
URL quotient & .00020 [.00000,.00060] & .00000 [.00000,.00000] & .00000 [.00000,.00000] & .00000 [.00000,.00000] & .00000 & .000 [.000,.000] & 50\\
\textbf{Lineage quotient} & \textbf{.00000 [.00000,.00000]} & \textbf{.00000 [.00000,.00000]} & \textbf{.00000 [.00000,.00000]} & .00880 [.00710,.01066] & \textbf{$1.13\!\times\!10^{-8}$} & \textbf{.000 [.000,.000]} & 50\\
\bottomrule
\end{tabular}%
}
\end{table*}
\subsection{AVeriTeC Profile Construction}
We use the AVeriTeC train and development sets \cite{schlichtkrull2023averitec}. We retain Supported and Refuted claims that have non-empty, non-Unanswerable QA evidence linked to a source URL, and require at least two usable reports per claim. Before stance filtering, these rules give \(2{,}226\) eligible claims and \(7{,}051\) candidate QA reports.

For each report, the NLI premise contains the question and answer, followed by the explanation when one is available. The hypothesis is the claim. A frozen \texttt{facebook/bart-large-mnli} model \cite{lewis2020bart,williams2018mnli} maps entailment to support, contradiction to attack, and neutral to discard. A non-neutral report is retained when its predicted class probability is at least \(0.60\). That probability is used as the report strength after clipping to \([0.05,0.95]\). The topic base strength is \(0.5\), and model inputs are truncated to \(1{,}024\) tokens.

The model-assigned stance labels and the resulting profiles were independently cross-checked by two human verifiers. Within each gold AVeriTeC label, eligible claims are ordered by a fixed pseudorandom permutation with seed \(2027\), then truncated to at most \(150\) claims after requiring at least two retained reports. Only \(127\) Supported claims satisfy the final filters. The final audit therefore contains \(277\) profiles, with \(127\) Supported and \(150\) Refuted claims, and \(711\) reports in total.

\begin{table}[t]
\caption{AVeriTeC profile statistics after the stated filters.}
\label{tab:real-flow}
\centering
\footnotesize
\setlength{\tabcolsep}{3pt}
\begin{tabular}{lr@{\quad}lr}
\toprule
Statistic & Value & Statistic & Value\\
\midrule
Claims & 277 & Supported / Refuted & 127 / 150\\
Reports & 711 & Support / attack & 152 / 559\\
Reports per claim & 2.57 & Mixed-sign claims & 47\\
Content-Q changed & 3 & URL-Q changed & 2\\
Fresh positions & 711 & FASR-eligible claims & 47\\
\bottomrule
\end{tabular}
\end{table}

\subsection{Controlled Lineage, Canonicalization and Uncertainty}
AVeriTeC does not provide evidential-root identifiers. We therefore assign each original QA report a distinct controlled root. The replica intervention adds seven exact copies of the selected report and preserves its content, URL, sign, strength, and root. The matched fresh intervention preserves the same judged fields but assigns a new root. The experiment consequently evaluates aggregation under supplied lineage rather than provenance discovery.

Canonical content is obtained by Unicode NFKC normalization case folding, and whitespace normalization of the QA content. The URL quotient collapses reports only when canonical content and sign agree, and the source URLs have the same canonical form. URL canonicalization lower-cases the host, removes \texttt{www.} and common tracking parameters, normalizes ports and slashes, and sorts the remaining query parameters.

Confidence intervals in Table~\ref{tab:real-full} use \(10{,}000\) claim-clustered bootstrap resamples stratified by AVeriTeC label. All interventions for a single claim are resampled together. MinG is the observed minimum over evaluated report positions rather than a bootstrap estimate. FASR uses the judge-specific eligible set.

\subsection{Sensitivity Analyses}
The lineage-quotient audit is repeated at NLI thresholds of \(0.50\) and \(0.70\), and every retained report is assigned a fixed strength of \(0.5\). Each sensitivity condition uses \(2{,}000\) claim-clustered bootstrap resamples. CD, RSS$_8$, RFR$_8$, and FASR are zero in every row, so Table~\ref{tab:real-sensitivity} reports only the quantities that vary.

\begin{table}[h]
\caption{Sensitivity of the AVeriTeC lineage-quotient audit. FLG is reported with a claim-clustered 95\% bootstrap interval.}
\label{tab:real-sensitivity}
\centering
\footnotesize
\setlength{\tabcolsep}{2.5pt}
\begin{tabular}{lrrrr}
\toprule
Condition & Claims & Pos. & FLG [95\% CI] & MinG\\
\midrule
NLI thr. .50 & 277 & 737 & .00833 [.00670,.01018] & $1.13{\times}10^{-8}$\\
NLI thr. .60 & 277 & 711 & .00880 [.00710,.01066] & $1.13{\times}10^{-8}$\\
NLI thr. .70 & 237 & 601 & .00655 [.00510,.00819] & $1.36{\times}10^{-8}$\\
Fixed str. .50 & 277 & 711 & .05132 [.04746,.05526] & .001953\\
\bottomrule
\end{tabular}
\end{table}

Across these conditions, the controlled replica remains exactly neutral, and every evaluated fresh root has a positive signed gain. The sensitivity analysis therefore checks that the reported separation is not specific to the \(0.60\) NLI threshold. It does not evaluate the NLI model as a fact checker.

\subsection{Reproduction and Scope Notes}
The AVeriTeC audit is determined by the dataset splits, the frozen NLI model, the filters and stance mapping above, the confidence threshold, strength clipping, topic-based strength, claim-selection rule, controlled root assignment, matched interventions, judge definitions, canonicalization rules, and bootstrap procedure.

Three scope restrictions matter when interpreting the audit. First, the root partition is supplied by construction because native AVeriTeC has no evidential-root labels. Second, only \(47\) profiles contain both support and attack reports, which limits the denominator of opposing-report multiplicity attacks. Third, both matched interventions keep canonical content fixed, so these experiments do not empirically separate the content-keyed quotient from the root--role lineage quotient under excerpt splitting. That distinction is established analytically in the main paper.

\section{AI-Assisted Preparation Disclosure}
In this work, we used generative AI tools solely to polish the manuscript's writing. All AI-assisted text was reviewed and revised by the authors. We take responsibility for the final content of this work, including any text produced with the aid of generative AI.

\end{document}

%% file: fig/fig_overview.tex
% =====================================================================
%  Figure 1 -- overview (ML-conference style, full width: figure*)
%  Requires: tikz with libraries arrows.meta, calc  (already in main.tex)
%  Usage in main.tex:
%     \begin{figure*}[t]
%       \centering
%       \input{fig_overview}
%       \caption{...}\label{fig:overview}
%     \end{figure*}
% =====================================================================
\definecolor{ink}{HTML}{2B2D42}
\definecolor{mute}{HTML}{6B7080}
\definecolor{panelbg}{HTML}{F6F7F9}
\definecolor{panelln}{HTML}{DADDE3}
\definecolor{rootA}{HTML}{4C72B0}   % evidential root l1
\definecolor{rootB}{HTML}{DD8452}   % evidential root l2
\definecolor{rootC}{HTML}{8172B2}   % evidential root l3
\definecolor{good}{HTML}{2A9D8F}    % lineage quotient / correct behaviour
\definecolor{bad}{HTML}{C44E52}     % failure / occurrence judge
\definecolor{neutralc}{HTML}{8C8F99}% content quotient
\begingroup

\begin{tikzpicture}[
  x=1cm, y=1cm,
  font=\fontsize{6.2}{7.4}\selectfont,
  >={Latex[length=1.5mm,width=1.1mm]},
  panel/.style={fill=panelbg, draw=panelln, line width=0.4pt, rounded corners=4pt},
  ttl/.style={anchor=west, text=ink, font=\bfseries\fontsize{7.6}{9}\selectfont},
  num/.style={circle, fill=ink, text=white, inner sep=0pt, minimum size=3.5mm,
              font=\bfseries\fontsize{6.5}{7}\selectfont},
  small/.style={font=\fontsize{5.4}{6.4}\selectfont, text=mute},
  chip/.style={draw=#1, fill=#1!9, line width=0.55pt, rounded corners=2pt,
               minimum width=1.45cm, minimum height=0.31cm, text width=1.15cm,
               align=left, inner xsep=2.6pt, inner ysep=0.5pt, anchor=west, text=ink},
  tag/.style={fill=#1, text=white, rounded corners=1.2pt, inner sep=0.8pt,
              minimum height=0.21cm, font=\bfseries\fontsize{5.2}{6}\selectfont},
  agent/.style={circle, draw=ink!55, fill=white, line width=0.45pt, inner sep=0pt,
                minimum size=0.29cm, font=\fontsize{5}{6}\selectfont, text=ink},
  claim/.style={draw=ink!70, fill=ink!8, line width=0.55pt, rounded corners=2pt,
                minimum width=0.46cm, minimum height=0.42cm, font=\fontsize{7}{8}\selectfont},
  flow/.style={->, line width=0.5pt, draw=#1},
  pill/.style={draw=panelln, fill=white, rounded corners=2.5pt, text width=#1,
               align=left, anchor=south west, inner xsep=4.2pt, inner ysep=2.4pt,
               text=ink, font=\fontsize{5.9}{7.1}\selectfont,
               execute at begin node={\hyphenpenalty=10000\relax\exhyphenpenalty=10000\relax}},
  accent/.style={path picture={\fill[#1] (path picture bounding box.south west)
               rectangle ([xshift=1.7pt]path picture bounding box.north west);}},
  hdr/.style={align=center, text=ink, font=\bfseries\fontsize{6}{6.8}\selectfont},
  mchip/.style={draw=#1, fill=#1!9, line width=0.5pt, rounded corners=1.4pt,
                minimum width=0.62cm, minimum height=0.2cm, inner sep=0pt},
  mtag/.style={fill=#1, text=white, rounded corners=0.8pt, inner sep=0.5pt,
               minimum height=0.13cm, font=\bfseries\fontsize{4.3}{5}\selectfont},
  gbox/.style={draw=ink!55, fill=panelbg, line width=0.5pt, rounded corners=1.5pt,
               inner sep=1.8pt, text=ink, font=\fontsize{6}{7}\selectfont},
  cond/.style={draw=#1!60, fill=#1!10, line width=0.45pt, rounded corners=1.5pt,
               inner sep=1.8pt, text=ink, font=\fontsize{5.6}{6.5}\selectfont},
  act/.style={minimum width=0.29cm, minimum height=0.2cm, rounded corners=1.2pt, inner sep=0pt,
               line width=0.45pt, text=ink, font=\bfseries\fontsize{4.6}{5}\selectfont},
  gn/.style={circle, draw=ink!60, fill=white, line width=0.5pt, inner sep=0pt, minimum size=0.25cm,
             font=\fontsize{5.2}{6}\selectfont, text=ink},
  cell/.style={minimum width=0.56cm, minimum height=0.15cm, inner sep=0pt, text=ink,
               font=\fontsize{4.7}{5}\selectfont},
  pics/doc/.style={code={
      \fill[#1!14] (-0.15,-0.19) -- (-0.15,0.19) -- (0.06,0.19) -- (0.15,0.10) -- (0.15,-0.19) -- cycle;
      \draw[#1, line width=0.55pt, line join=round]
           (-0.15,-0.19) -- (-0.15,0.19) -- (0.06,0.19) -- (0.15,0.10) -- (0.15,-0.19) -- cycle;
      \draw[#1, line width=0.45pt, line join=round] (0.06,0.19) -- (0.06,0.10) -- (0.15,0.10);
      \draw[#1!75, line width=0.4pt] (-0.09,0.05) -- (0.08,0.05) (-0.09,-0.03) -- (0.08,-0.03)
           (-0.09,-0.11) -- (0.03,-0.11);}},
  pics/ok/.style={code={
      \fill[good] (0,0) circle (0.125);
      \draw[white, line width=0.85pt, line cap=round, line join=round]
           (-0.058,0.002) -- (-0.016,-0.046) -- (0.062,0.052);}},
  pics/no/.style={code={
      \fill[bad] (0,0) circle (0.125);
      \draw[white, line width=0.85pt, line cap=round]
           (-0.045,-0.045) -- (0.045,0.045) (-0.045,0.045) -- (0.045,-0.045);}},
]

% card: #1,#2 = lower-left, #3,#4 = upper-right, #5 (unused, kept for compatibility), #6 title, #7 reference
\def\ovcard#1#2#3#4#5#6#7{%
  \fill[white, rounded corners=2.5pt] (#1,#2) rectangle (#3,#4);
  \draw[panelln, line width=0.4pt, rounded corners=2.5pt] (#1,#2) rectangle (#3,#4);
  \node[anchor=north west, inner sep=0pt, text=ink, font=\bfseries\fontsize{6.2}{7}\selectfont]
       at ($(#1,#4)+(0.17,-0.08)$) {#6};
  \node[anchor=north east, inner sep=0pt, text=mute, font=\fontsize{5.4}{6.4}\selectfont]
       at ($(#3,#4)+(-0.1,-0.09)$) {#7};
}
% ---------------------------------------------------------------------
% Panel frames
% ---------------------------------------------------------------------
\def\H{6.8}
\fill[panel] (0,0)    rectangle (6.0,\H);
\fill[panel] (6.3,0)  rectangle (11.9,\H);
\fill[panel] (12.2,0) rectangle (17.8,\H);
\draw[panel, fill=none] (0,0)    rectangle (6.0,\H);
\draw[panel, fill=none] (6.3,0)  rectangle (11.9,\H);
\draw[panel, fill=none] (12.2,0) rectangle (17.8,\H);

% =====================================================================
% (a) Same visible agreement, different evidence
% =====================================================================
\begin{scope}[shift={(0,0)}]
  \node[num] at (0.32,6.45) {a};
  \node[ttl] at (0.52,6.45) {Same agreement, different evidence};

  % ---- row labels
  \node[rotate=90, font=\bfseries\fontsize{6.2}{7}\selectfont, text=ink] at (0.26,5.1) {Replication};
  \node[rotate=90, font=\bfseries\fontsize{6.2}{7}\selectfont, text=ink] at (0.26,3.25) {Corroboration};

  % ---- Replication row: one root, three renderings
  \pic at (0.92,5.1) {doc=rootA};
  \node[font=\fontsize{6.5}{7}\selectfont, text=rootA] at (0.92,4.72) {$\ell_1$};
  \foreach \y/\a/\t in {5.55/$A_1$/report, 5.1/$A_2$/relay, 4.65/$A_2'$/excerpt}{
     \node[agent] (ag) at (1.42,\y) {\a};
     \node[chip=rootA] (ch) at (1.62,\y) {\t};
     \node[tag=rootA, anchor=east] at ([xshift=-1.2pt]ch.east) {$\ell_1$};
     \draw[flow=rootA!80] (1.09,5.1) -- (ag.west);
  }
  \node[claim] (cR) at (3.82,5.1) {$c$};
  \foreach \y/\d in {5.55/2.6pt, 5.1/0pt, 4.65/-2.6pt}
     \draw[flow=ink!55] (3.07,\y) -- ([yshift=\d]cR.west);
  \node[anchor=west, align=left, text=ink] at (4.14,5.1)
       {\textbf{1 root}\\[-0.5pt]{\color{mute}should add}\\[-0.5pt]{\color{mute}no weight}};
  \node[small, anchor=west] at (0.62,4.24)
       {O2--O5: relay, paraphrase, alias, excerpt split \,$\cdot$\, cost $c^r\!<c^f$};

  % ---- Corroboration row: three roots
  \foreach \y/\a/\r/\l in {3.7/$A_1$/rootA/$\ell_1$, 3.25/$A_2$/rootB/$\ell_2$, 2.8/$A_3$/rootC/$\ell_3$}{
     \pic at (0.92,\y) {doc=\r};
     \node[agent] (ag) at (1.42,\y) {\a};
     \node[chip=\r] (ch) at (1.62,\y) {report};
     \node[tag=\r, anchor=east] at ([xshift=-1.2pt]ch.east) {\l};
     \draw[flow=\r!80] (1.09,\y) -- (ag.west);
  }
  \node[claim] (cC) at (3.82,3.25) {$c$};
  \foreach \y/\d in {3.7/2.6pt, 3.25/0pt, 2.8/-2.6pt}
     \draw[flow=ink!55] (3.07,\y) -- ([yshift=\d]cC.west);
  \node[anchor=west, align=left, text=ink] at (4.14,3.25)
       {\textbf{3 roots}\\[-0.5pt]{\color{mute}should}\\[-0.5pt]{\color{mute}reinforce}};
  \node[small, anchor=west] at (0.62,2.39)
       {O1: fresh evidential root \,$\cdot$\, cost $c^f$};

  % ---- card: information boundary (Thm 1), drawn as a diagram
  \ovcard{0.15}{0.15}{5.85}{2.08}{bad}{Information boundary}{Thm.~\ref{thm:information}}
  % matched pair: P_R (replica of l1) and P_F (fresh root l2)
  \node[font=\fontsize{6.6}{7}\selectfont, text=ink] at (0.45,1.3) {$P_R$};
  \node[mchip=rootA] (r1) at (1.05,1.42) {}; \node[mtag=rootA, anchor=east] at ([xshift=-1pt]r1.east) {$\ell_1$};
  \node[mchip=rootA] (r2) at (1.05,1.18) {}; \node[mtag=rootA, anchor=east] at ([xshift=-1pt]r2.east) {$\ell_1$};
  \node[font=\fontsize{6.6}{7}\selectfont, text=ink] at (5.3,1.3) {$P_F$};
  \node[mchip=rootA] (f1) at (4.65,1.42) {}; \node[mtag=rootA, anchor=east] at ([xshift=-1pt]f1.east) {$\ell_1$};
  \node[mchip=rootB] (f2) at (4.65,1.18) {}; \node[mtag=rootB, anchor=east] at ([xshift=-1pt]f2.east) {$\ell_2$};
  % interface image: root labels stripped -> identical
  \node[fill=panelbg, draw=panelln, rounded corners=1.5pt, minimum width=0.95cm, minimum height=0.6cm] (img) at (2.85,1.3) {};
  \foreach \yy in {1.42,1.18} \fill[ink!25, rounded corners=1pt] (2.54,\yy-0.07) rectangle (3.16,\yy+0.07);
  \draw[->, ink!55, line width=0.5pt] (1.42,1.3) -- node[above, inner sep=1pt, font=\fontsize{6}{6}\selectfont] {$\varphi$} (img.west);
  \draw[->, ink!55, line width=0.5pt] (4.28,1.3) -- node[above, inner sep=1pt, font=\fontsize{6}{6}\selectfont] {$\varphi$} (img.east);
  \node[font=\fontsize{5.8}{6.5}\selectfont, text=ink] at (2.85,0.9) {$\varphi(P_R)=\varphi(P_F)$};
  \node[small] at (2.85,0.69) {holds for every observation interface $\varphi$};
  % the two demands that cannot both hold
  \node[cond=good, anchor=west, font=\fontsize{5.1}{6}\selectfont] (cl) at (0.32,0.38) {RRI: $S(P_R){=}S(P)$};
  \node[cond=good, anchor=east, font=\fontsize{5.1}{6}\selectfont] (cr) at (5.72,0.38) {FLR: $S(P_F){>}S(P)$};
  \pic at (2.47,0.38) {no};
  \node[anchor=west, inner sep=0pt, text=bad, font=\bfseries\fontsize{5.6}{6.5}\selectfont] at (2.64,0.38) {not both};
  \draw[bad!70, line width=0.45pt, dash pattern=on 1.2pt off 0.9pt] (cl.east) -- (2.33,0.38);
  \draw[bad!70, line width=0.45pt, dash pattern=on 1.2pt off 0.9pt] (3.3,0.38) -- (cr.west);
\end{scope}

% =====================================================================
% (b) What must the judge retain?
% =====================================================================
\begin{scope}[shift={(6.3,0)}]
  \node[num] at (0.32,6.45) {b};
  \node[ttl] at (0.52,6.45) {What must the judge retain?};

  % column headers
  \node[hdr] at (3.30,5.82) {Replica};
  \node[hdr] at (4.15,5.82) {Excerpt\\split};
  \node[hdr] at (5.00,5.82) {Fresh\\root};
  \node[small, anchor=west, font=\itshape\fontsize{5.4}{6.4}\selectfont] at (0.18,5.32) {desired effect on score};
  \foreach \x/\t in {3.30/ignore, 4.15/ignore, 5.00/count}
     \node[small, font=\itshape\fontsize{5.4}{6.4}\selectfont] at (\x,5.32) {\t};
  \draw[panelln, line width=0.5pt] (0.18,5.13) -- (5.42,5.13);

  % highlight band for the lineage quotient
  \fill[good!11, rounded corners=2pt] (0.1,2.93) rectangle (5.5,3.47);
  \draw[good!60, line width=0.5pt, rounded corners=2pt] (0.1,2.93) rectangle (5.5,3.47);

  % rows: name, key, marks
  \foreach \y/\n/\k/\ma/\mb/\mc in {
      4.85/Occurrence multiset/every report kept/no/no/ok,
      4.30/Content quotient $q_C$/{key $(\kappa,p,s)$}/ok/no/no,
      3.75/Content-keyed $q_K$/{key $(\kappa,\lambda,p,s)$}/ok/no/ok}{
     \node[anchor=west, align=left, text=ink] at (0.18,\y)
          {\textbf{\n}\\[-0.6pt]{\color{mute}\fontsize{5.4}{6.2}\selectfont\k}};
     \pic at (3.30,\y) {\ma};
     \pic at (4.15,\y) {\mb};
     \pic at (5.00,\y) {\mc};
  }
  \node[anchor=west, align=left, text=ink] at (0.18,3.2)
       {\textbf{Lineage quotient $q_L$} {\color{good}\textbf{(ours)}}\\[-0.6pt]{\color{mute}\fontsize{5.4}{6.2}\selectfont key $(\lambda,p,s)$, strength $\max\tau$}};
  \pic at (3.30,3.2) {ok};
  \pic at (4.15,3.2) {ok};
  \pic at (5.00,3.2) {ok};

  \node[small, anchor=west, align=left] at (0.18,2.6)
       {\tikz[baseline=-0.6ex]\pic{no}; in replica or excerpt $\Rightarrow$ profitable replication (Prop.~\ref{prop:profitable})\\[0.8pt]
        \tikz[baseline=-0.6ex]\pic{no}; in fresh root $\Rightarrow$ corroboration erased (Prop.~\ref{prop:contentcollapse})};

  % ---- quotient mini-illustration
  \begin{scope}[shift={(0,2.0)}]
    \node[small, anchor=west] at (0.18,0.0) {$P$:};
    \fill[rootA] (0.60,0) circle (0.072);
    \fill[rootA] (0.81,0) circle (0.072);
    \fill[rootA] (1.02,-0.07) -- (1.095,0.07) -- (0.945,0.07) -- cycle; % excerpt
    \fill[rootA] (1.23,0) circle (0.072);
    \fill[rootB] (1.44,0) circle (0.072);
    \draw[->, line width=0.7pt, good] (1.64,0) -- (2.26,0)
         node[midway, above=-0.5pt, font=\fontsize{6}{6}\selectfont, text=good] {$q_L$};
    \node[fill=rootA!15, draw=rootA, rounded corners=1.5pt, inner sep=1.2pt,
          font=\fontsize{5.2}{6}\selectfont, text=ink] at (2.76,0) {unit\,$\cdot$\,$\max\tau$};
    \node[fill=rootB!15, draw=rootB, rounded corners=1.5pt, inner sep=1.2pt,
          font=\fontsize{5.2}{6}\selectfont, text=ink] at (3.74,0) {unit\,$\cdot\,\tau$};
    \node[small, anchor=west, align=left] at (4.14,0) {distinct roots\\[-1pt]stay apart};
    \node[small, anchor=west] at (0.46,-0.28) {drops wording, identity, multiplicity \& root names};
  \end{scope}

  % ---- card: factorization + minimality (Thms 2-3), drawn as a commutative diagram
  \ovcard{0.15}{0.15}{5.45}{1.52}{good}{Coarsest sufficient key}{Thms.~\ref{thm:factor}--\ref{thm:minimal}}
  \node[gbox] (nP) at (0.62,0.92) {$P$};
  \node[gbox, draw=good, fill=good!12] (nQ) at (2.3,0.92) {$q_L(P)$};
  \node[gbox] (nS) at (3.5,0.92) {$[0,1]$};
  \node[gbox, draw=ink!40, dash pattern=on 1.5pt off 1pt] (nF) at (2.3,0.37) {$\varphi(P)$};
  \draw[->, good, line width=0.7pt] (nP) -- node[above, inner sep=1pt, font=\fontsize{5.8}{6}\selectfont] {$q_L$} (nQ);
  \draw[->, ink!70, line width=0.55pt] (nQ) -- node[above, inner sep=1pt, font=\fontsize{5.8}{6}\selectfont] {$\widetilde J$} (nS);

  \draw[->, ink!45, line width=0.5pt, dash pattern=on 1.5pt off 1pt] (nP.south east) -- node[below left, inner sep=0pt, font=\fontsize{5.8}{6}\selectfont] {$\varphi$} (nF.west);
  \draw[->, good, line width=0.55pt] (nF.north) -- node[right, inner sep=2pt, font=\fontsize{5.8}{6}\selectfont] {$g$} (nQ.south);
  \node[anchor=west, align=left, text=ink, font=\fontsize{5.8}{6.8}\selectfont] at (3.9,0.94)
       {$J=\widetilde J\circ q_L$\\[-1pt]{\color{mute}\fontsize{5}{6}\selectfont every RRI\,+\,LNI $J$}};
  \node[anchor=west, align=left, text=ink, font=\fontsize{5.8}{6.8}\selectfont] at (2.9,0.43)
       {$q_L=g\circ\varphi$ \;{\color{mute}\fontsize{5}{6}\selectfont $\forall$ sufficient $\varphi$}};
\end{scope}

% =====================================================================
% (c) Lifting to argumentation graphs
% =====================================================================
\begin{scope}[shift={(12.2,0)}]
  \node[num] at (0.32,6.45) {c};
  \node[ttl] at (0.52,6.45) {Lifting to argumentation graphs};

  % ---- occurrence QBAF  ->  quotient QBAF
  \node[anchor=south, inner sep=1pt, text=ink, font=\bfseries\fontsize{5.8}{6.6}\selectfont] at (1.35,5.92) {occurrence QBAF $Q^o$};
  \node[anchor=south, inner sep=1pt, text=ink, font=\bfseries\fontsize{5.8}{6.6}\selectfont] at (4.3,5.92) {quotient $Q^o/{\equiv_R}$};
  \fill[white, draw=panelln, line width=0.4pt, rounded corners=2.5pt] (0.2,3.42) rectangle (2.55,5.9);
  \fill[white, draw=good!70, line width=0.5pt, rounded corners=2.5pt] (3.1,3.42) rectangle (5.45,5.9);
  \draw[->, good, line width=0.9pt] (2.62,4.66) -- (3.03,4.66);
  \node[anchor=south, inner sep=1pt, text=good, font=\fontsize{6}{6.5}\selectfont] at (2.82,4.7) {$/{\equiv_R}$};
  \node[anchor=north, inner sep=1pt, text=mute, font=\fontsize{4.8}{5.5}\selectfont, align=center] at (2.82,4.6) {merge\\root--role\\classes};

  % left graph
  \begin{scope}[shift={(0,0)}]
    \node[claim, minimum width=0.36cm, minimum height=0.3cm, font=\fontsize{6}{6.5}\selectfont] (c) at (1.45,5.42) {$c$};
    \node[claim, minimum width=0.36cm, minimum height=0.3cm, font=\fontsize{6}{6.5}\selectfont] (d) at (0.62,4.55) {$d$};
    \node[circle, fill=rootA!30, draw=rootA, line width=0.5pt, inner sep=0pt, minimum size=0.21cm] (o1) at (1.18,4.25) {};
    \node[circle, fill=rootA!30, draw=rootA, line width=0.5pt, inner sep=0pt, minimum size=0.21cm] (o2) at (1.5,4.25) {};
    \node[inner sep=0pt, minimum size=0.22cm] (o3) at (1.82,4.25) {};
    \fill[rootA!30, draw=rootA, line width=0.5pt, line join=round] (1.82,4.36) -- (1.93,4.16) -- (1.71,4.16) -- cycle;
    \node[circle, fill=rootB!30, draw=rootB, line width=0.5pt, inner sep=0pt, minimum size=0.21cm] (o4) at (2.28,4.55) {};
    \node[circle, fill=rootC!30, draw=rootC, line width=0.5pt, inner sep=0pt, minimum size=0.21cm] (o5) at (0.45,3.75) {};
    \node[circle, fill=rootC!30, draw=rootC, line width=0.5pt, inner sep=0pt, minimum size=0.21cm] (o6) at (0.79,3.75) {};
    \foreach \o in {o1,o2,o3} \draw[->, good, line width=0.5pt] (\o) -- (c);
    \draw[->, good, line width=0.5pt] (d) -- (c);
    \draw[->, bad, line width=0.5pt] (o4) -- (c);
    \draw[->, good, line width=0.5pt] (o5) -- (d);
    \draw[->, good, line width=0.5pt] (o6) -- (d);
    \node[small, anchor=north, inner sep=1pt, text=rootA] at (1.5,4.1) {$\ell_1\times3$};
    \node[small, anchor=north, inner sep=1pt, text=rootB] at (2.28,4.41) {$\ell_2$};
    \node[small, anchor=west, inner sep=1pt, text=rootC] at (0.92,3.75) {$\ell_3\times2$};
  \end{scope}

  % right graph
  \begin{scope}[shift={(2.9,0)}]
    \node[claim, minimum width=0.36cm, minimum height=0.3cm, font=\fontsize{6}{6.5}\selectfont] (c) at (1.45,5.42) {$c$};
    \node[claim, minimum width=0.36cm, minimum height=0.3cm, font=\fontsize{6}{6.5}\selectfont] (d) at (0.62,4.55) {$d$};
    \node[circle, fill=rootA!30, draw=rootA, line width=0.7pt, inner sep=0pt, minimum size=0.27cm] (u1) at (1.5,4.25) {};
    \node[circle, fill=rootB!30, draw=rootB, line width=0.5pt, inner sep=0pt, minimum size=0.21cm] (u2) at (2.28,4.55) {};
    \node[circle, fill=rootC!30, draw=rootC, line width=0.7pt, inner sep=0pt, minimum size=0.27cm] (u3) at (0.62,3.75) {};
    \draw[->, good, line width=0.6pt] (u1) -- (c);
    \draw[->, good, line width=0.5pt] (d) -- (c);
    \draw[->, bad, line width=0.5pt] (u2) -- (c);
    \draw[->, good, line width=0.6pt] (u3) -- (d);
    \node[small, anchor=west, inner sep=1pt, text=rootA] at (1.6,4.13) {$\bar\tau{=}\max\tau$};
    \node[small, anchor=north, inner sep=1pt, text=rootB] at (2.28,4.41) {$\ell_2$};
    \node[small, anchor=west, inner sep=1pt, text=rootC] at (0.8,3.75) {$\ell_3$};
  \end{scope}

  % legend
  \begin{scope}[shift={(0.25,3.2)}]
    \draw[->, good, line width=0.5pt] (0,0) -- (0.3,0);
    \node[small, anchor=west, text=ink, inner sep=1.5pt] (lg1) at (0.3,0) {support};
    \draw[->, bad, line width=0.5pt] ([xshift=5pt]lg1.east) -- ++(0.3,0) node[small, anchor=west, text=ink, inner sep=1.5pt] (lg2) {attack};
    \node[claim, minimum width=0.24cm, minimum height=0.2cm, inner sep=0pt, anchor=west] (lg3) at ([xshift=6pt]lg2.east) {};
    \node[small, anchor=west, text=ink, inner sep=1.5pt] (lg4) at (lg3.east) {claim, never merged};
    \fill[ink!25] ([xshift=7pt]lg4.east) circle (0.07);
    \node[small, anchor=west, text=ink, inner sep=1.5pt] at ([xshift=10pt]lg4.east) {report};
  \end{scope}

  % ---- generic lifting, as an equation with its three properties
  \ovcard{0.15}{1.9}{5.45}{2.98}{good}{Generic lifting}{Thm.~\ref{thm:lifting}}
  \node[text=ink, font=\fontsize{7.2}{8}\selectfont] at (2.8,2.52) {$\mathrm{LQ}[\sigma](Q^o)\;=\;\sigma\bigl(Q^o/{\equiv_R}\bigr)$};
  \node[cond=good, font=\fontsize{5}{6}\selectfont] at (1.05,2.13) {structural RRI};
  \node[cond=good, font=\fontsize{5}{6}\selectfont] at (2.35,2.13) {LNI};
  \node[cond=good, font=\fontsize{5}{6}\selectfont] at (3.95,2.13) {$=\sigma$ on singleton classes};

  \begin{scope}[yshift=-1.87cm]
  % ---- card: replication-proofness (Prop 1, Cor 1) as a trace edit
  \ovcard{0.15}{2.02}{2.58}{3.62}{good}{Replication-proof}{}
  \node[small, anchor=north west, inner sep=0pt] at (0.32,3.29) {Prop.~\ref{prop:repproof} $\cdot$ Cor.~\ref{cor:repproof}};
  \node[act, draw=rootA, fill=rootA!18] at (0.42,2.95) {O1};
  \node[act, draw=ink!40, fill=ink!6, dash pattern=on 1pt off 0.8pt] at (0.72,2.95) {O2};
  \node[act, draw=ink!40, fill=ink!6, dash pattern=on 1pt off 0.8pt] at (1.02,2.95) {O5};
  \node[act, draw=rootB, fill=rootB!18] at (1.32,2.95) {O1};
  \draw[->, ink!60, line width=0.5pt] (1.5,2.95) -- node[above, inner sep=0.6pt, font=\fontsize{4.4}{5}\selectfont, text=mute] {drop} (1.84,2.95);
  \node[act, draw=rootA, fill=rootA!18] at (2.03,2.95) {O1};
  \node[act, draw=rootB, fill=rootB!18] at (2.33,2.95) {O1};
  \pic at (0.44,2.4) {ok};
  \node[anchor=west, align=left, inner sep=0pt, text=ink, font=\fontsize{5.3}{6.1}\selectfont] at (0.62,2.4)
       {same verdict\\at every $\theta$};
  % utility bars: dropping replication raises U_i by k c^r
  \fill[ink!28] (1.9,2.16) rectangle (2.08,2.34);
  \fill[good]   (2.18,2.16) rectangle (2.36,2.5);
  \draw[ink!60, line width=0.35pt] (1.84,2.16) -- (2.42,2.16);
  \node[anchor=south, inner sep=0.8pt, text=good, font=\fontsize{5.2}{6}\selectfont] at (2.27,2.5) {$+kc^r$};

  % ---- card: fresh roots still count (Thms 6-7) as a signed path
  \ovcard{2.7}{2.02}{5.45}{3.62}{rootA}{Fresh roots count}{}
  \node[small, anchor=north west, inner sep=0pt] at (2.87,3.29) {Thms.~\ref{thm:propagate}--\ref{thm:modulus}};
  \node[gn, draw=rootB, fill=rootB!25] (t) at (3.05,2.95) {$t$};
  \node[gn] (d) at (3.72,2.95) {$d$};
  \node[gn] (e) at (4.42,2.95) {$e$};
  \node[claim, minimum width=0.32cm, minimum height=0.28cm, font=\fontsize{5.8}{6}\selectfont] (cc) at (5.12,2.95) {$c$};
  \draw[->, good, line width=0.55pt] (t) -- node[above, inner sep=0.6pt, font=\bfseries\fontsize{6}{6}\selectfont] {$+$} (d);
  \draw[->, bad,  line width=0.55pt] (d) -- node[above, inner sep=0.6pt, font=\bfseries\fontsize{6}{6}\selectfont] {$-$} (e);
  \draw[->, bad,  line width=0.55pt] (e) -- node[above, inner sep=0.6pt, font=\bfseries\fontsize{6}{6}\selectfont] {$-$} (cc);
  \node[small, inner sep=0pt, font=\fontsize{5}{6}\selectfont] at (3.72,2.71) {sign-consistent path};
  \node[inner sep=0pt, text=good, font=\fontsize{5.2}{6}\selectfont] at (5.12,3.25) {$\Delta\sigma\!>\!0$};
  \node[fill=rootA!8, draw=rootA!45, rounded corners=1.4pt, inner sep=1pt, text=ink,
        font=\fontsize{4.6}{6}\selectfont] at (4.07,2.33) {$\Delta\sigma\ge t(1{-}v)\min\{v_0,1{-}v_0\}$};
  \end{scope}
\end{scope}

\end{tikzpicture}
\endgroup

%% file: temp/07_Exp_New.tex
\section{Experiments and Results}
\label{sec:audits}

Our experiments test whether a judge ignores repeated evidence while responding strictly to a matched fresh root. Each audit starts with one root per report and applies the paired interventions of Def.~\ref{def:matched}. We measure score and verdict changes under repetition, signed fresh-root gains, and preservation of clean scores. Supplied lineage isolates aggregation behavior.

\subsection{A Toy Debate}
\label{sec:toy}
 
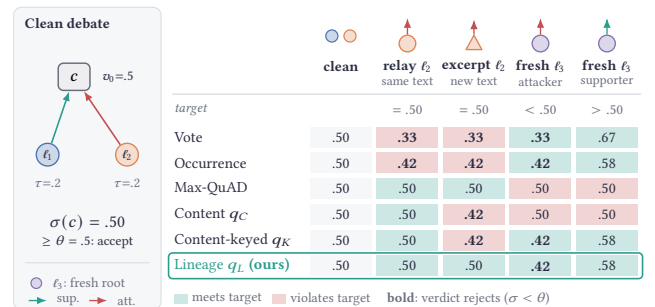
\begin{figure}[b]
\centering
\input{fig/toy_illus}
\caption{Only the lineage quotient separates replicas from fresh roots. Toy debate of Ex.~\ref{ex:deviation} under \dfquad{} ($v_0=.5$, $\tau=.2$, $\theta=.5$), one added report per column. Cells give the topic score, colored by whether the judge meets the target (replicas neutral, fresh roots move the score in their sign).}
\Description{An argumentation graph with one supporter and one attacker of a claim, beside a table of topic scores for six judges under four single-report interventions: relay, excerpt, fresh attacker, and fresh supporter.}
\label{fig:toy}
\end{figure}
 
Fig.~\ref{fig:toy} replays the profile of Ex.~\ref{ex:deviation}: one supporter from root $\ell_1$ and one attacker from $\ell_2$, both of strength $0.2$, with topic base $0.5$, so every judge scores $0.5$ and accepts at $\theta=0.5$. We then add a single report. A relay of $\ell_2$ lowers occurrence \dfquad{} to $0.42$ and reverses the verdict, while the lineage quotient stays at $0.5$ (\RRI). A matched attacker from a fresh root $\ell_3$ gives occurrence \dfquad{} the same $0.42$: without lineage, the replica and the fresh root are indistinguishable, as Thm.~\ref{thm:information} requires. The lineage quotient separates them. It falls to $0.42$ for the fresh attacker and rises to $0.58$ for a matched fresh supporter, and both changes equal the bound $t(1-v)\min\{v_0,1-v_0\}=0.2\cdot0.8\cdot0.5=0.08$ of Thm.~\ref{thm:modulus}, so the modulus is attained. Max-QuAD and the content quotient ignore both fresh reports, while an excerpt of $\ell_2$ moves both content-based quotients but not the lineage quotient. The audits below test relays and fresh roots at scale, with content held fixed.

\subsection{Exhaustive Structural Audit}

We repeat the test across direct acyclic star QBAFs, where every report supports or attacks the topic. Topic and leaf strengths range over $G=\{0.2,0.4,0.6,0.8\}$, with one to three supporters and one to three attackers. Exhaustive enumeration yields 28,224 profiles, 153,216 report positions, and 306,432 matched interventions at threshold $0.5$. All strengths lie in $(0,1)$, making every fresh intervention \dfquad-eligible.

Five judges receive identical profiles: \emph{Vote} (fraction of supporting occurrences), \emph{occurrence \dfquad} (noisy-OR over all occurrences), \emph{Max-QuAD} (maximum aggregation), the \emph{content quotient}, and the \emph{lineage quotient} (\lqdf). For report $i$, let $T_8^i(P)$ add seven replicas, giving eight occurrences in total, and $F^i(P)$ add one matched fresh-root report. With $s_i\in\{+1,-1\}$ denoting report sign:
\begin{gather*}
\mathrm{RSS}_8
=\mathbb{E}\left|S(T_8^i(P))-S(P)\right|,
\quad
\mathrm{RFR}_8
=\mathbb{E}\,\mathbf{1}\!\left[D(T_8^i(P))\ne D(P)\right],
\\
\mathrm{FLG}
=\mathbb{E}\,s_i\left(S(F^i(P))-S(P)\right),
\quad
\mathrm{MinG}
=\min_{P,i}\,s_i\left(S(F^i(P))-S(P)\right).
\end{gather*}
RSS$_8$ and RFR$_8$ measure the mean absolute score change and verdict-flip rate under repetition. FLG is the mean signed fresh-root gain, and MinG its minimum over tested positions. FASR is the verdict-flip rate from eightfold repetition of an opposing report, on each judge's eligible set. CD is the mean absolute deviation from occurrence \dfquad{} on clean profiles.

\begin{table}[t]
\caption{Exhaustive structural audit. Target: $\mathrm{CD}=\mathrm{RSS}_8=\mathrm{RFR}_8=\mathrm{FASR}=0$ and $\mathrm{MinG}>0$. Bold marks joint satisfaction.}
\label{tab:audit}
\centering
\footnotesize
\setlength{\tabcolsep}{2.3pt}
\begin{tabular}{lcccccc}
\toprule
Judge & CD & RSS$_8$ & RFR$_8$ & FLG & MinG & FASR\\
\midrule
Vote & .1818 & .2723 & .4582 & .0752 & .0500 & 1.0000\\
Occurrence \dfquad & .0000 & .0677 & .0932 & .0298 & .0013 & .1922\\
Max-QuAD & .0448 & .0000 & .0000 & .0000 & .0000 & .0000\\
Content quotient & .0000 & .0000 & .0000 & .0000 & .0000 & .0000\\
\textbf{Lineage quotient} & \textbf{.0000} & \textbf{.0000} & \textbf{.0000} & .0298 & \textbf{.0013} & \textbf{.0000}\\
\bottomrule
\end{tabular}
\end{table}

Table~\ref{tab:audit} shows the toy illustration's two failure modes across the grid. Occurrence \dfquad{} responds strictly to fresh evidence (MinG $=.0013$), but repetition shifts its score by $.0677$ on average, flips $9.32\%$ of verdicts, and flips $19.22\%$ of 74,338 eligible opposing-report cases. Vote is more sensitive (RSS$_8=.2723$, RFR$_8=.4582$, FASR $=1$). Max-QuAD and the content quotient ignore repetition but have MinG $=0$: matched fresh evidence can have no effect. Among the tested judges, only the lineage quotient meets the joint target: CD, RSS$_8$, RFR$_8$, and FASR are zero, while FLG $=.0298$ and MinG $=.0013$. The latter is the rounded analytic minimum $0.00128$ of the qFLR modulus in Thm.~\ref{thm:modulus}.

We also exhaustively enumerate 2,052 chain, fork, and diamond DAGs with strengths in $\{0.2,0.5,0.8\}$. Across 2,700 replicated positions and 4,104 fresh interventions, the lineage lifting has zero maximum clean deviation and replication shift, with minimum signed fresh-root gain $1.28\times10^{-4}>0$. Fresh reports attach directly to the topic; intermediate-node fresh interventions are not audited.

\subsection{Real-profile Audit on AVeriTeC}
\label{sec:realaudit}

We apply the same test to source-linked question and answer evidence from Supported and Refuted AVeriTeC claims \cite{schlichtkrull2023averitec}. A frozen BART-MNLI classifier \cite{lewis2020bart,williams2018mnli} assigns support, attack, or discard, yielding 277 profiles (127 Supported, 150 Refuted) and 711 report positions. Two human verifiers independently cross-check the assigned stances and profiles. Each original report receives a distinct controlled root. Repetition preserves it, while the matched fresh report changes only the root field. We also evaluate a \emph{URL quotient}, which merges matching canonical content and sign when canonicalized URLs agree. Confidence intervals use claim-clustered bootstrap resampling. Full preprocessing and sensitivity analyses appear in the supplement.

\begin{table}[t]
\caption{Matched audit on 277 AVeriTeC-derived profiles. Same target as Table~\ref{tab:audit}.}
\label{tab:real-audit}
\centering
\footnotesize
\setlength{\tabcolsep}{1.8pt}
\begin{tabular}{lcccccc}
\toprule
Judge & CD & RSS$_8$ & RFR$_8$ & FLG & MinG & FASR\\
\midrule
Vote & .0187 & .0572 & .0703 & .0219 & .0000 & 1.0000\\
Occurrence \dfquad & .0000 & .0114 & .0689 & .0088 & $1.1\!\times\!10^{-8}$ & .9800\\
Max-QuAD & .0278 & .0000 & .0000 & .0000 & .0000 & .0000\\
Content quotient & .0003 & .0000 & .0000 & .0000 & .0000 & .0000\\
URL quotient & .0002 & .0000 & .0000 & .0000 & .0000 & .0000\\
\textbf{Lineage quotient} & \textbf{.0000} & \textbf{.0000} & \textbf{.0000} & .0088 & \textbf{$1.1\!\times\!10^{-8}$} & \textbf{.0000}\\
\bottomrule
\end{tabular}
\end{table}

Table~\ref{tab:real-audit} reproduces the same separation on evidence-derived profiles. Occurrence \dfquad{} remains fresh-responsive (FLG $=.0088$, MinG $=1.1\times10^{-8}$), but repetition shifts its score by $.0114$ on average, flips $6.89\%$ of verdicts, and succeeds in 49 of 50 eligible opposing-report cases (FASR $=.98$). Vote has positive mean fresh-root gain ($.0219$) but MinG $=0$, so mean response does not imply strict response everywhere. Max-QuAD and both content-based quotients ignore repetition but have MinG $=0$. Among the tested judges, only the lineage quotient meets the joint target: CD, RSS$_8$, RFR$_8$, and FASR are zero, and all 711 fresh contributions have positive signed gains (mean $.0088$, 95\% CI $[.0071,.0107]$). The small minimum is consistent with the $(1-v)$ saturation factor in Thm.~\ref{thm:modulus}.

%% file: fig/toy_illus.tex
% =====================================================================
%  Toy debate figure (single column). Requires tikz + arrows.meta, calc.
%  Scores are generated by toy_check.py (DF-QuAD, v0 = .5, tau = .2, theta = .5).
% =====================================================================
\definecolor{ink}{HTML}{2B2D42}%
\definecolor{mute}{HTML}{6B7080}%
\definecolor{panelbg}{HTML}{F6F7F9}%
\definecolor{panelln}{HTML}{DADDE3}%
\definecolor{rootA}{HTML}{4C72B0}%
\definecolor{rootB}{HTML}{DD8452}%
\definecolor{rootC}{HTML}{8172B2}%
\definecolor{good}{HTML}{2A9D8F}%
\definecolor{bad}{HTML}{C44E52}%
\begingroup%
\begin{tikzpicture}[
  x=1cm, y=1cm,
  font=\fontsize{6}{7}\selectfont,
  >={Latex[length=1.3mm,width=1mm]},
  small/.style={font=\fontsize{5.2}{6.2}\selectfont, text=mute},
  rep/.style={circle, draw=#1, fill=#1!28, line width=0.5pt, inner sep=0pt, minimum size=0.3cm,
              font=\fontsize{5.2}{6}\selectfont, text=ink},
  claim/.style={draw=ink!70, fill=ink!8, line width=0.55pt, rounded corners=2pt,
                minimum width=0.42cm, minimum height=0.36cm, font=\fontsize{6.5}{7}\selectfont},
  tcell/.style={minimum width=0.82cm, minimum height=0.27cm, inner sep=0pt, text=ink},
  hdr/.style={font=\bfseries\fontsize{5.5}{6.4}\selectfont, text=ink, inner sep=0pt},
  sub/.style={font=\fontsize{5}{6}\selectfont, text=mute, inner sep=0pt},
]
% ---------------- clean debate ----------------
\fill[panelbg, draw=panelln, line width=0.4pt, rounded corners=3pt] (0,0.1) rectangle (1.95,4.1);
\node[hdr, anchor=north west] at (0.14,3.97) {Clean debate};
\node[claim] (c) at (0.8,3.18) {$c$};
\node[small, anchor=west, text=ink] at (1.06,3.18) {$v_0{=}.5$};
\node[rep=rootA] (s) at (0.45,2.15) {$\ell_1$};
\node[rep=rootB] (a) at (1.5,2.15) {$\ell_2$};
\draw[->, good, line width=0.6pt] (s) -- (c);
\draw[->, bad,  line width=0.6pt] (a) -- (c);
\node[small, anchor=north] at (0.45,1.95) {$\tau{=}.2$};
\node[small, anchor=north] at (1.5,1.95) {$\tau{=}.2$};
\node[text=ink, font=\fontsize{6}{7}\selectfont] at (0.975,1.25) {$\sigma(c)=.50$};
\node[small, text=ink] at (0.975,1.0) {$\ge\theta=.5$: accept};
\draw[panelln, line width=0.4pt] (0.15,0.72) -- (1.8,0.72);
\fill[rootC!28, draw=rootC, line width=0.45pt] (0.3,0.45) circle (0.07);
\node[small, anchor=west] at (0.4,0.45) {$\ell_3$: fresh root};
\draw[->, good, line width=0.5pt] (0.2,0.22) -- (0.44,0.22);
\node[small, anchor=west] at (0.46,0.22) {sup.};
\draw[->, bad, line width=0.5pt] (1.0,0.22) -- (1.24,0.22);
\node[small, anchor=west] at (1.26,0.22) {att.};

% ---------------- column headers with glyphs of the added report ----------------
% clean
\fill[rootA!28, draw=rootA, line width=0.45pt] (4.21,3.72) circle (0.08);
\fill[rootB!28, draw=rootB, line width=0.45pt] (4.45,3.72) circle (0.08);
% relay of l2 (same text)
\node[rep=rootB, minimum size=0.22cm] (h1) at (5.21,3.62) {};
\draw[->, bad, line width=0.5pt] (h1.north) -- ++(0,0.22);
% excerpt of l2 (new text)
\fill[rootB!28, draw=rootB, line width=0.5pt, line join=round] (6.09,3.73) -- (6.195,3.54) -- (5.985,3.54) -- cycle;
\draw[->, bad, line width=0.5pt] (6.09,3.74) -- ++(0,0.2);
% fresh attacker l3
\node[rep=rootC, minimum size=0.22cm] (h3) at (6.97,3.62) {};
\draw[->, bad, line width=0.5pt] (h3.north) -- ++(0,0.22);
% fresh supporter l3
\node[rep=rootC, minimum size=0.22cm] (h4) at (7.85,3.62) {};
\draw[->, good, line width=0.5pt] (h4.north) -- ++(0,0.22);

\node[hdr] at (4.33,3.32) {clean};
\node[hdr] at (5.21,3.32) {relay $\ell_2$};
\node[hdr] at (6.09,3.32) {excerpt $\ell_2$};
\node[hdr] at (6.97,3.32) {fresh $\ell_3$};
\node[hdr] at (7.85,3.32) {fresh $\ell_3$};
\node[sub] at (5.21,3.12) {same text};
\node[sub] at (6.09,3.12) {new text};
\node[sub] at (6.97,3.12) {attacker};
\node[sub] at (7.85,3.12) {supporter};
\draw[panelln, line width=0.5pt] (2.1,2.95) -- (8.3,2.95);
\node[anchor=west, inner sep=0pt, font=\itshape\fontsize{5.2}{6}\selectfont, text=mute] at (2.12,2.74) {target};
\foreach \x/\t in {5.21/{$=.50$}, 6.09/{$=.50$}, 6.97/{$<.50$}, 7.85/{$>.50$}}
  \node[font=\fontsize{5.4}{6}\selectfont, text=mute] at (\x,2.74) {\t};

% highlight band for the lineage quotient
\draw[good, line width=0.6pt, rounded corners=1.5pt] (2.05,0.51) rectangle (8.32,0.85);

% ---------------- score cells (generated) ----------------
  \node[anchor=west, inner sep=0pt, text=ink, font=\fontsize{5.8}{6.6}\selectfont] at (2.12,2.38) {Vote};
  \node[tcell, fill=panelbg, font=\fontsize{5.8}{6.6}\selectfont] at (4.33,2.38) {.50};
  \node[tcell, fill=bad!24, font=\bfseries\fontsize{5.8}{6.6}\selectfont] at (5.21,2.38) {.33};
  \node[tcell, fill=bad!24, font=\bfseries\fontsize{5.8}{6.6}\selectfont] at (6.09,2.38) {.33};
  \node[tcell, fill=good!24, font=\bfseries\fontsize{5.8}{6.6}\selectfont] at (6.97,2.38) {.33};
  \node[tcell, fill=good!24, font=\fontsize{5.8}{6.6}\selectfont] at (7.85,2.38) {.67};
  \node[anchor=west, inner sep=0pt, text=ink, font=\fontsize{5.8}{6.6}\selectfont] at (2.12,2.04) {Occurrence};
  \node[tcell, fill=panelbg, font=\fontsize{5.8}{6.6}\selectfont] at (4.33,2.04) {.50};
  \node[tcell, fill=bad!24, font=\bfseries\fontsize{5.8}{6.6}\selectfont] at (5.21,2.04) {.42};
  \node[tcell, fill=bad!24, font=\bfseries\fontsize{5.8}{6.6}\selectfont] at (6.09,2.04) {.42};
  \node[tcell, fill=good!24, font=\bfseries\fontsize{5.8}{6.6}\selectfont] at (6.97,2.04) {.42};
  \node[tcell, fill=good!24, font=\fontsize{5.8}{6.6}\selectfont] at (7.85,2.04) {.58};
  \node[anchor=west, inner sep=0pt, text=ink, font=\fontsize{5.8}{6.6}\selectfont] at (2.12,1.7) {Max-QuAD};
  \node[tcell, fill=panelbg, font=\fontsize{5.8}{6.6}\selectfont] at (4.33,1.7) {.50};
  \node[tcell, fill=good!24, font=\fontsize{5.8}{6.6}\selectfont] at (5.21,1.7) {.50};
  \node[tcell, fill=good!24, font=\fontsize{5.8}{6.6}\selectfont] at (6.09,1.7) {.50};
  \node[tcell, fill=bad!24, font=\fontsize{5.8}{6.6}\selectfont] at (6.97,1.7) {.50};
  \node[tcell, fill=bad!24, font=\fontsize{5.8}{6.6}\selectfont] at (7.85,1.7) {.50};
  \node[anchor=west, inner sep=0pt, text=ink, font=\fontsize{5.8}{6.6}\selectfont] at (2.12,1.36) {Content $q_C$};
  \node[tcell, fill=panelbg, font=\fontsize{5.8}{6.6}\selectfont] at (4.33,1.36) {.50};
  \node[tcell, fill=good!24, font=\fontsize{5.8}{6.6}\selectfont] at (5.21,1.36) {.50};
  \node[tcell, fill=bad!24, font=\bfseries\fontsize{5.8}{6.6}\selectfont] at (6.09,1.36) {.42};
  \node[tcell, fill=bad!24, font=\fontsize{5.8}{6.6}\selectfont] at (6.97,1.36) {.50};
  \node[tcell, fill=bad!24, font=\fontsize{5.8}{6.6}\selectfont] at (7.85,1.36) {.50};
  \node[anchor=west, inner sep=0pt, text=ink, font=\fontsize{5.8}{6.6}\selectfont] at (2.12,1.02) {Content-keyed $q_K$};
  \node[tcell, fill=panelbg, font=\fontsize{5.8}{6.6}\selectfont] at (4.33,1.02) {.50};
  \node[tcell, fill=good!24, font=\fontsize{5.8}{6.6}\selectfont] at (5.21,1.02) {.50};
  \node[tcell, fill=bad!24, font=\bfseries\fontsize{5.8}{6.6}\selectfont] at (6.09,1.02) {.42};
  \node[tcell, fill=good!24, font=\bfseries\fontsize{5.8}{6.6}\selectfont] at (6.97,1.02) {.42};
  \node[tcell, fill=good!24, font=\fontsize{5.8}{6.6}\selectfont] at (7.85,1.02) {.58};
  \node[anchor=west, inner sep=0pt, text=good, font=\fontsize{5.8}{6.6}\selectfont] at (2.12,0.68) {Lineage $q_L$ \textbf{(ours)}};
  \node[tcell, fill=panelbg, font=\fontsize{5.8}{6.6}\selectfont] at (4.33,0.68) {.50};
  \node[tcell, fill=good!24, font=\fontsize{5.8}{6.6}\selectfont] at (5.21,0.68) {.50};
  \node[tcell, fill=good!24, font=\fontsize{5.8}{6.6}\selectfont] at (6.09,0.68) {.50};
  \node[tcell, fill=good!24, font=\bfseries\fontsize{5.8}{6.6}\selectfont] at (6.97,0.68) {.42};
  \node[tcell, fill=good!24, font=\fontsize{5.8}{6.6}\selectfont] at (7.85,0.68) {.58};

% ---------------- legend ----------------
\fill[good!24] (2.12,0.18) rectangle (2.3,0.3);
\node[small, anchor=west, inner sep=1pt] at (2.32,0.24) {meets target};
\fill[bad!24] (3.42,0.18) rectangle (3.6,0.3);
\node[small, anchor=west, inner sep=1pt] at (3.62,0.24) {violates target};
\node[small, anchor=west, inner sep=1pt] at (4.9,0.24) {\textbf{bold}: verdict rejects ($\sigma<\theta$)};
\end{tikzpicture}%
\endgroup